\documentclass{article}

\usepackage{microtype}
\usepackage{graphicx}
\usepackage{subfigure}
\usepackage{booktabs} 

\usepackage{hyperref}

\usepackage[accepted]{icml2023}

\usepackage{amsmath}
\usepackage{amssymb}
\usepackage{mathtools}
\usepackage{amsthm}

\usepackage[capitalize,noabbrev]{cleveref}

\theoremstyle{plain}
\newtheorem{theorem}{Theorem}[section]
\newtheorem{proposition}[theorem]{Proposition}
\newtheorem{lemma}[theorem]{Lemma}
\newtheorem{corollary}[theorem]{Corollary}
\theoremstyle{definition}
\newtheorem{definition}[theorem]{Definition}

\theoremstyle{remark}

\usepackage[textsize=tiny]{todonotes}

\usepackage{tikz}
\usepackage{graphicx}
\usepackage{nicefrac}
\usepackage{aliascnt}
\usepackage{booktabs}
\usepackage{multirow}

\newenvironment{CompactItemize}{
\begin{list}{$\bullet$}{%
\setlength{\leftmargin}{10pt}
\setlength{\itemindent}{2pt}
\setlength{\topsep}{-1pt}
\setlength{\itemsep}{-1pt}
}}
{\end{list}}

\newaliascnt{fact}{theorem}

\aliascntresetthe{fact}
\crefname{fact}{Fact}{Facts}

\newaliascnt{conjecture}{theorem}

\aliascntresetthe{conjecture}
\crefname{conjecture}{Conjecture}{Conjectures}

\newaliascnt{claim}{theorem}

\aliascntresetthe{claim}
\crefname{claim}{Claim}{Claims}

\newaliascnt{question}{theorem}

\aliascntresetthe{question}
\crefname{question}{Question}{Questions}

\newaliascnt{exercise}{theorem}

\aliascntresetthe{exercise}
\crefname{exercise}{Exercise}{Exercises}

\newaliascnt{example}{theorem}

\aliascntresetthe{example}
\crefname{example}{Example}{Examples}

\newaliascnt{notation}{theorem}

\aliascntresetthe{notation}
\crefname{notation}{Notation}{Notations}

\newaliascnt{problem}{theorem}

\aliascntresetthe{problem}
\crefname{problem}{Problem}{Problems}

\newcommand{\norm}[1]{\lVert#1\rVert}

\def\E{\mathbb E}
\newcommand{\Var}{{\bf Var}}

\newcommand{\R}{\mathbb R}

\icmltitlerunning{Fast Private Kernel Density Estimation via Locality Sensitive Quantization}

\begin{document}

\twocolumn[
\icmltitle{Fast Private Kernel Density Estimation via Locality Sensitive Quantization}




\begin{icmlauthorlist}
\icmlauthor{Tal Wagner}{amazon}
\icmlauthor{Yonatan Naamad}{amazon}
\icmlauthor{Nina Mishra}{amazon}
\end{icmlauthorlist}

\icmlaffiliation{amazon}{Amazon}

\icmlcorrespondingauthor{Tal Wagner}{tal.wagner@gmail.com}

\icmlkeywords{Machine Learning, ICML}

\vskip 0.3in
]



\printAffiliationsAndNotice{}  


\begin{abstract}
We study efficient mechanisms for differentially private kernel density estimation (DP-KDE). 
Prior work for the Gaussian kernel described algorithms that run in time exponential in the number of dimensions $d$.  This paper breaks the exponential barrier, and shows how the KDE can privately be approximated in time linear in $d$, making it feasible for high-dimensional data. 
We also present improved bounds for low-dimensional data.

Our results are obtained through a general framework, which we term Locality Sensitive Quantization (LSQ), for constructing private KDE mechanisms where existing KDE approximation techniques can be applied. 
It lets us leverage several efficient non-private KDE methods---like Random Fourier Features, the Fast Gauss Transform, and Locality Sensitive Hashing---and ``privatize'' them in a black-box manner. 
Our experiments demonstrate that our resulting DP-KDE mechanisms are fast and accurate on large datasets in both high and low dimensions.
\end{abstract}

\section{Introduction}

Private analysis of massive-scale data is a prominent current challenge in computing and machine learning. 
On the one hand, it is widely acknowledged that big datasets drive advances and progress in many important problem spaces. 
On the other hand, when the data contains sensitive information such as personal or medical details, it is often necessary to preserve the privacy of individual dataset records. Scalable methods for private computations are therefore crucial for progress in medical, financial and many other domains.

Differential privacy (DP) \cite{dwork2006calibrating} is a rigorous and powerful notion of privacy-preserving computation, widely accepted in machine learning. 
Unfortunately, it often comes at a high computational cost, and many DP computations are dramatically less efficient than their non-private counterparts. This makes them infeasible for use on data of the size and dimensionality that matches present-day scale.

\textbf{DP-KDE.} 
In this paper we focus on private density estimation, a fundamental problem with numerous applications in data analysis and machine learning.
A popular way to convert a collection of data points into a smoothed probability distribution is the kernel density method, wherein a certain probability measure---say, a Gaussian---is centered at each data point, and the mixture of these measures is formed over the space. 
The kernel density estimate (KDE) at every point $y$ is the mean of all such Gaussians evaluated at $y$. 
This method has a long history in statistics and machine learning (e.g., \cite{shawe2004kernel,hofmann2008kernel}). Under private computation, it has recently been used for private crowdsourcing and  location sharing \cite{huai2019privacy,cunningham2021privacy}.

The associated algorithmic task is: given a dataset $X\subset\R^d$,
return a map $\hat e_X:\R^d\rightarrow\R$ that approximates the KDE map $y\mapsto\tfrac1{|X|}\sum_{x\in X}e^{-\norm{x-y}_2^2/\sigma^2}$. In DP-KDE, $\hat e_X$ must also be private w.r.t.~$X$, no matter how many times it is queried.

Absent privacy limitations, the Gaussian KDE at a point $y$ can be evaluated in time $O(nd)$, where $n$ is the number of data points and $d$ is their dimension. Many efficient approximation methods have been developed to speed this up even further for large-scale data (e.g., \cite{greengard1991fast,rahimi2007random,charikar2017hashing,phillips2020near}.
In sharp contrast, in the DP setting, existing methods for privately estimating the Gaussian KDE have running time exponential in $d$ \cite{hall2013differential,hall2013new,wang2016differentially,alda2017bernstein}. 
This makes them infeasible in many important cases where KDE is utilized in high-dimensional feature spaces.

In this paper, we close this gap by systematically studying efficient mechanisms for DP-KDE, and obtain improved results in both the high and low dimensional regimes.

\subsection{Our Results}

\begin{table*}
\caption{\small $\epsilon$-DP KDE function release mechanisms for the Gaussian kernel, that satisfy $(\alpha,\eta$)-approximation (\Cref{def:addapx}). The dataset contains $n$ points in $\R^d$. 
Where $n$ appears in the curator time, note that it must be at least as large as the sample complexity.
(*) SmallDB, PMW and LSQ-FGT assume that all points lie in a ball of radius $\Phi$. (**) EvenTrig and Bernstein assume that all points lie in $[-1,1]^d$, and their performance depends on the bandwidth $\sigma$ under this scaling. ($\dagger$) For EvenTrig, $\eta=\exp(-\Omega(n^{d/(2d+O(\sigma^2))} ))$.}
{\renewcommand{\arraystretch}{2}
\begin{center}
\begin{tabular}{llcccl}
\toprule
 & \sc Mechanism & \sc Curator time & \sc Sample complexity & \\
\midrule
\multirow{4}{*}{Prior}  &  SmallDB  & $\exp\left(\frac{\min\{\sqrt{d\log(1/\alpha)},1/\alpha\} \cdot d^2\log^2(\Phi/\alpha)}{\epsilon \cdot \alpha^2}\right)$ & $O\left(\frac{\min\{\sqrt{d\log(1/\alpha)},1/\alpha\} \cdot d\log(\Phi/\alpha)}{\epsilon \cdot \alpha^2}\right)$ & \scriptsize{(*)} \\
&  PMW &  $\tilde O\left(d\cdot\left(\Phi/\alpha\right)^d\right)$ & $\tilde O\left( \frac{d^2\log^2(\Phi/\alpha)}{\epsilon \cdot \alpha^3}\right)$ & \scriptsize{(*)} \\
& EvenTrig & $O\left(2^d+dn^{1+d/(2d+\Theta(\sigma^2))}\right)$ & $O\left(\frac{1}{(\epsilon \cdot \alpha)^{1+\Theta(d/\sigma^2)}}\right)$ & \scriptsize{(**), ($\dagger$)} \\
& Bernstein & $O\left(dn\cdot\left(2^d+(\tfrac{\epsilon \cdot n}{\log(1/\eta)})^{d/(d+\Theta(\sigma^2))}\right)\right)$ & $O\left(\frac{\log(1/\eta)}{\epsilon \cdot \alpha^{1+\Theta(d/\sigma^2)}}\right)$ & \scriptsize{(**)} \\
\midrule
\multirow{2}{*}{Ours} &  LSQ-RFF &  $O(dn\cdot\frac{\log(1/\eta)}{\alpha^2})$ & $O(\frac{\log(1/\eta)}{\epsilon \cdot \alpha^2})$  & \\
&  LSQ-FGT & $(dn+(\frac{\Phi}{\sqrt{d}})^d) \cdot O(\log(1/\alpha))^{d}\cdot \log(1/\eta)$ & $\frac{\log(1/\eta)}{\epsilon \cdot \alpha} \cdot (\log(1/\alpha))^{O(d)}$ & \scriptsize{(*)} \\
\bottomrule
\end{tabular}
\end{center}}
\label{tbl:main}
\end{table*}

We focus on the Gaussian kernel, although we will discuss other kernels as well. 
Our first result is an $\epsilon$-DP function release mechanism for Gaussian KDE (see \Cref{sec:kde,sec:dp} for formal definitions), whose running time is only linear in $d$, making it suitable for high-dimensional data. 

\begin{theorem}[Gaussian DP-KDE in high dimensions]\label{thm:gauhighdim}
There is an $\epsilon$-DP function release mechanism for $(\alpha,\eta)$-approximation of Gaussian KDE on datasets in $\R^d$ of size $n\geq O(\log(1/\eta)/(\epsilon\alpha^2))$, and:
\begin{CompactItemize} 
  \item The curator runs in time $O(nd\log(1/\eta)/\alpha^2)$.
  \item The output size is $O(d\log(1/\eta)/\alpha^2)$.
  \item The client runs in time $O(d\log(1/\eta)/\alpha^2)$. 
\end{CompactItemize}
\end{theorem}


Our second result is a Gaussian DP-KDE mechanism for low-dimensional data,  if the points reside in a bounded region. It improves the dependence on $\alpha$ to nearly linear.

\begin{theorem}[Gaussian DP-KDE in low dimensions]\label{thm:gaulowdim}
There is an $\epsilon$-DP function release mechanism for $(\alpha,\eta)$-approximation of Gaussian KDE on datasets in $\R^d$ of size $n\geq\log(1/\eta) \cdot (\log(1/\alpha))^{O(d)}/(\epsilon\alpha)$ and that are contained in a ball of radius $\Phi$, and:
\begin{CompactItemize} 
  \item The curator runs in time $(nd+(\frac{\Phi}{\sqrt{d}})^d) \cdot O(\log(1/\alpha))^{O(d)}\cdot \log(1/\eta)$.
  \item The output size is $O((1+\frac{\Phi}{\sqrt{d}})(\log(1/\alpha)))^d\log(1/\eta)$.
  \item The client runs in time $(\log(1/\alpha))^{O(d)}\log(1/\eta)$. 
\end{CompactItemize}
\end{theorem}


\textbf{Our approach.}
We obtain our results by introducing a framework we call \emph{locality sensitive quantization} (LSQ). 
It captures a certain type of KDE approximation algorithms, which are based on point quantization. 
On the one hand, we show that the LSQ properties are by themselves sufficient to imply an efficient DP-KDE mechanism. 
On the other hand, we show that many popular approximation methods for KDE already possess these properties---including random Fourier features (RFF) \cite{rahimi2007random}, the Fast Gauss Transform (FGT) \cite{greengard1991fast}, and locality sensitive hashing (LSH) \cite{charikar2017hashing}. Thus, by plugging each of these methods into the LSQ framework, we immediately get efficient DP-KDE mechanisms for the kernels they approximate. 


The key properties highlighted in the LSQ framework are \emph{quantization} (i.e., the dataset is quantized into a small number of values), \emph{range} (these values are bounded), and \emph{sparsity} (each single point affects only a small number of values). As mentioned, several non-private KDE algorithms operate in this manner, as it can lead to efficient and accurate approximation. The reason it is also useful for efficient DP mechanisms is roughly that quantization lets us add noise to a compact representation of the data, saving time; bounded range means the noise can have small magnitude; and sparsity ensures the noise does not add up too much. 

On a broader conceptual level, there is a fundamental connection between DP and non-private approximation algorithms based on sketching (or quantization). Indeed, many recent works have uncovered such connections \cite{blocki2012johnson,feldman2021lossless,aumuller2021differentially,coleman2020one,pagh2022improved,nikolov2022private}. 
Our work adds to this growing line of research.

\subsection{Preliminaries: Kernel Density Estimation (KDE)}\label{sec:kde}

A \emph{kernel} is a function $k:\R^d\times\R^d\rightarrow[0,1]$ that measures similarity between points in $\R^d$. Popular kernels include:
\begin{CompactItemize}
  \item Gaussian kernel: $k(x,y)=\exp(-\norm{x-y}_2^2/\sigma^2)$
  \item Laplacian kernel: $k(x,y)=\exp(-\norm{x-y}_1/\sigma)$
  \item Cauchy kernel: $k(x,y)=\prod_{j=1}^d2/(1+(x_j-y_j)^2/\sigma^2)$
\end{CompactItemize}
Here, $\sigma>0$ is the \emph{bandwidth} parameter. For simplicity, we set $\sigma=1$ throughout; this does not limit generality, as we can scale the point coordinates accordingly.

Let $X\subset\R^d$ be a finite dataset. 
The \emph{kernel density estimation (KDE)} map $KDE_X:\R^d\rightarrow[0,1]$ is defined as
\[  KDE_X(y)=\frac1{|X|}\sum_{x\in X}k(x,y) . \]
Our goal will be to approximate the KDE map in the following formal sense.
\begin{definition}\label{def:addapx}
Let $\hat e:\R^d\rightarrow[0,1]$ be a randomized mapping, and let $\alpha,\eta \in (0,1)$. 
We say that $\hat e$ is an \emph{$(\alpha,\eta)$-approximation} for $KDE_X$ if for every $y\in\R^d$,
\[
  \Pr[|\hat e(y) - KDE_X(y)| \leq \alpha] \geq 1-\eta .
\]
\end{definition}


\subsection{Preliminaries: Differential Privacy (DP)}\label{sec:dp}
Differential privacy \cite{dwork2006calibrating} is a setting that involves communication between two parties: the \emph{curator}, who holds a dataset $X$, and the \emph{client}, who wishes to obtain the result of some computation on the dataset. We say that two datasets $X,X'$ are \emph{neighboring} if omitting a single data point from one of them yields the other. 

\begin{definition}\label{def:dp}
Let $M$ be a randomized algorithm (called a \emph{mechanism}) that maps an input dataset to a range of outputs $\mathcal O$. For $\epsilon,\delta>0$, we say that $M$ is \emph{$(\epsilon,\delta)$-DP if} for every neighboring datasets $X,X'$ and every $O\subset\mathcal O$, 
\[ \Pr[M(X)\in O] \leq \exp(\epsilon)\cdot\Pr[M(X')\in O]+\delta . \]
The case $\delta=0$ is called \emph{pure} differential privacy, and in that case we say that $M$ is \emph{$\epsilon$-DP}.
\end{definition}
In this paper we focus on pure differential privacy---given a desired privacy level $\epsilon>0$, the curator is only allowed to release the results of $\epsilon$-DP computations on $X$. See \Cref{sec:apxdp} for a discussion of non-pure DP-KDE. 

\textbf{Function release.}
We focus on the differentially private function release communication model. In this model, the curator releases an $\epsilon$-DP description of a function $\hat e(\cdot)$ that satisfies \cref{def:addapx} for the dataset $X$, without seeing any queries in advance. The client can then use this description to compute $\hat e(y)$ for any query $y$. Note that since $\hat e(\cdot)$ itself is $\epsilon$-DP, the client can use it for infinitely many queries without compromising the privacy of the dataset. However, as more queries are computed, the overall number of inaccurate estimates is expected to grow (as only an expected $(1-\eta)$-fraction of them is guaranteed to have error within $\pm\alpha$).

\textbf{Sample complexity.} There is an inherent trade-off between privacy and approximation (or utility). It can be expressed as the minimal dataset size for which both are simultaneously possible---a quantity known as the \emph{sample complexity}.
Intuitively, the larger the dataset is, the easier it is to maintain the privacy of each point while releasing accurate global information.
Formally, given $\epsilon,\alpha,\eta>0$, the sample complexity $\mathrm{sc}(M)$ of a mechanism $M$ is the smallest $s$ such that $M$ is both $\epsilon$-DP and satisfies $(\alpha,\eta)$-approximation for all datasets of size at least $s$. 

The sample complexity affects the running time: 
On the one hand, the dataset size $n$ must be at least $\mathrm{sc}(M)$.
On the other hand, since the KDE at any query point is the mean of values in $[0,1]$, the curator can initially subsample the dataset down to size $O(\log(1/\eta)/\alpha^2)$ while maintaining $(\alpha,\eta)$-approximation, by Hoeffding's inequality. The upshot is that w.l.o.g., $n$ can always be assumed to satisfy $\mathrm{sc}(M) \leq n \leq O(\max\{\mathrm{sc}(M), \log(1/\eta)/\alpha^2\})$.


\textbf{Computational efficiency.}
In addition to privacy and utility, we also want the curator and client algorithms to be time-efficient, and the curator output size to be space-efficient.

\subsection{Prior Work}\label{sec:prior}

\textbf{Generic linear queries.} 
KDE queries belong to a broader class of linear queries, which are extensively studied in the DP literature. Two classical mechanisms for them are SmallDB \cite{blum2013learning} and Private Multiplicative Weights (PMW) \cite{hardt2010multiplicative,gupta2012iterative}. These mechanisms are designed for the DP \emph{query release} model, where the curator only releases responses to queries provided by the client. 
Nonetheless they can be adapted to the more general function release model, if the KDE problem is restricted to points contained in a ball of radius $\Phi$. 
We provide more details on this transformation in \cref{sec:linearappendix}. In either the query or function release model,
these mechanisms run in time at least exponential in $d$.



\textbf{DP-KDE in low dimensions.}
Several authors explored mechanisms specifically for DP-KDE. 
\cite{hall2013differential} presented a non-pure DP mechanism, based on noise correlation, in the query release model.
However, when used for function release, its running time is exponential in $d$ (see \Cref{sec:apxdp} for details).  
\cite{wang2016differentially} introduced an $\epsilon$-DP function release mechanism for $(\alpha,\eta)$-approximation of smooth functions, assuming all points lie in $[-1,1]^d$, using a basis of even trigonometric polynomials. Its performance for DP-KDE depends the bandwidth $\sigma$ (under scaling the data into $[-1,1]^d$), and has a fixed value for $\eta$. It also entails computations that do not admit a closed form and require numerical methods. 
\cite{alda2017bernstein} introduced the Bernstein mechanism, based on Bernstein basis polynomials, and obtained similar guarantees with any $\eta\in(0,1)$ and in a closed-form computation. 
The running time of both of these mechanisms is exponential in $d$. 



\textbf{Locality sensitive hashing (LSH).} 
Recently, \cite{coleman2020one} broke the $\exp(d)$ barrier for DP-KDE by using LSH \cite{indyk1998approximate}. The usefulness of LSH for non-private KDE has been observed in \cite{andoni2009nearest}, and recently regained much attention \cite{charikar2017hashing,siminelakis2019rehashing,coleman2020sub,backurs2019space,backurs2021faster}. Then, \cite{coleman2020one} showed it is also useful for DP-KDE. 
They obtained an $\epsilon$-DP mechanism with $(\alpha,\eta)$-approximation and running time only linear in $d$. 

However, their result does not apply to the Gaussian kernel. It is restricted to kernels that satisfy a property known as \emph{LSHability}, which roughly means they can be accurately described by LSH (see \Cref{sec:lsh} for the formal definition). 
While some popular kernels possess this property---perhaps most notably, the Laplacian kernel \cite{rahimi2007random,andoni2009nearest,backurs2019space}---other important kernels, like Gaussian and Cauchy, are not known nor believed to be LSHable (see, e.g., \cite{backurs2018efficient}).
See \Cref{sec:lshable} for specific LSHable kernels.


\textbf{Comparison to our results.}
The comparison is summarized in \Cref{tbl:main}. 
Our LSQ-RFF mechanism runs in time linear in $d$ and polynomial in $1/\alpha$. Its sample complexity and computational efficiency match those of \cite{coleman2020one}, but it works for a wider class of kernels. For the Gaussian kernel, it is the first to avoid an exponential dependence on $d$ in the running time. Furthermore, it does not require the data to be contained in a bounded region.
In the low-dimensional setting $d=O(1)$, our LSQ-FGT mechanism is the first to attain a nearly linear dependence of $O(\alpha^{-1}\mathrm{log}^{O(1)}(\alpha^{-1}))$ on the error $\alpha$.\footnote{Note that this is the dependence on $\alpha$ in both the sample complexity and the curator running time, since $n$ is lower bounded by the sample complexity.}


\textbf{Adaptive queries.}
The transformation of SmallDB and PMW from query release to function release, mentioned above, in fact endows them with a stronger property than \cref{def:addapx}: not only they succeed on every query with probability $1-\eta$, but they succeed on all queries \emph{simultaneously} with a fixed probability (say $0.9$). This enables the client to adaptively choose queries based on the results of previous queries, which is useful for data exploration, among other benefits (see \cite{cherapanamjeri2020adaptive}). The same transformation can be applied to our mechanisms as well; see \Cref{app:adaptive}.


\section{Locality Sensitive Quantization}\label{sec:lsq}

The following is the main definition for this paper.
\begin{definition}\label{def:lsq}
Let $Q,S\geq0$ be integers and $\alpha,R>0$. Let $\mathcal Q$ be a distribution over pairs $(f,g)$ such that:
\begin{CompactItemize}
  \item $f$ and $g$ are maps $f,g:\R^d\rightarrow[-R,R]^Q$.
  \item For every $x,y\in\R^d$, the $Q$-dimensional vectors $f(x)$ and $g(y)$ have each at most $S$ non-zero entries.
\end{CompactItemize}
We say that $\mathcal Q$ is an \emph{$\alpha$-approximate $(Q,R,S)$-locality sensitive quantization (abbrev.~$(Q,R,S)$-LSQ)} family for a kernel $k:\R^d\times\R^d\rightarrow[0,1]$, if for every $x,y\in\R^d$,
\[ \left|k(x,y) - \E_{(f,g)\sim\mathcal Q}[f(x)^Tg(y)] \right| \leq \alpha . \]
We call $k$ an \emph{$\alpha$-approximate $(Q,R,S)$-LSQable} kernel. 
If $\alpha=0$, we say that $\mathcal Q$ is an \emph{exact} $(Q,R,S)$-LSQ family for $k$, and that $k$ is \emph{$(Q,R,S)$-LSQable}.
\end{definition}

Intuitively, an LSQ family expresses the kernel as the expected inner product between vectors with a small number of entries $(Q)$, bounded range $(R)$, and bounded sparsity $(S)$. 
The definition is reminiscent of random features, Fast Multipole Methods  \cite{greengard1987fast}, and LSHability (\cref{def:lsh})---indeed, as we will see, it captures all of these. Its goal is to form an abstraction of the key properties that on the one hand ``automatically'' suffice for an efficient DP-KDE mechanism, and on the other hand are already shared by many prominent KDE methods.

\subsection{LSQ Mechanism for DP-KDE}\label{sec:lsq}


\newcommand{\INDENT}{\hspace{1em}}
\begin{algorithm}[tb]
\caption{: LSQ Mechanism for DP-KDE}
\label{alg:main}
\textbf{Curator}
\smallskip{\hrule height.2pt}
\begin{algorithmic}
   \STATE {\bfseries Input:} Dataset $X\subset\R^d$; $(Q,R,S)$-LSQ family $\mathcal Q$; privacy parameter $\epsilon>0$; integers $I\geq J>0$. 
   \STATE {\bfseries for} $i=1,\ldots,I$ {\bfseries do}
   \STATE\INDENT Sample $(f_i,g_i) \sim \mathcal Q$
   \STATE\INDENT $F_i\leftarrow \frac{1}{|X|}\sum_{x\in X}f_i(x)$ \hfill {\itshape // note: $F_i\in[-R,R]^Q$}
   \STATE\INDENT $\widetilde F_i \leftarrow F_i$ with an i.i.d. sample from $\mathrm{Laplace}(IRS/(\epsilon|X|))$ added to each coordinate
   \STATE {\bfseries release} $f_i,g_i,\widetilde F_i$ for all $i=1,\ldots,I$.
\end{algorithmic}
\smallskip{\hrule height.8pt}\smallskip
\textbf{Client}
\smallskip{\hrule height.2pt}
\begin{algorithmic}
   \STATE {\bfseries Input:} Query point $y\in\R^d$; the released $\{f_i,g_i,\widetilde F_i\}_{i=1}^I$.
   \STATE $I' \leftarrow \lfloor I/J\rfloor$
   \STATE {\bfseries for} $j=1,\ldots,J$ {\bfseries do}
   \STATE\INDENT $m_j \leftarrow \frac1{I'}\sum_{i=I'(j-1)+1}^{I'j}\widetilde F_i^Tg_i(y)$
   \STATE {\bfseries return} $\hat e(y) := \mathrm{median}(m_1,\ldots,m_J)$.
\end{algorithmic}
\end{algorithm}


Let $k$ be a kernel with an $\alpha$-approximate $(Q,R,S)$-LSQ family $\mathcal Q$.
The LSQ mechanism for DP-KDE is specified in \Cref{alg:main}. 
It is parameterized by the privacy level $\epsilon$, and by integers $I\geq J>0$ that govern the efficiency/utility trade-off. We discuss their role and how to set them in more detail in \cref{sec:mechanism_implementation}.
The formal properties of the mechanism are stated next, with proofs deferred to \cref{sec:lsqappendix}.


%

\begin{lemma}[privacy]\label{lmm:privacy}
    The mechanism is $\epsilon$-DP.
\end{lemma}

\begin{lemma}[efficiency]\label{lmm:efficiency}
Denote by $T_{\mathcal Q}$ the time to sample $(f,g)\sim\mathcal Q$, by $T_f,T_g$ the time to compute $f(x),g(y)$ given $x,y\in\R^d$ respectively, and by $L_{\mathcal Q}$ the description size in machine words of a pair $(f,g)$ sampled from $\mathcal Q$. Then,
\begin{CompactItemize}
    \item The curator runs in time $O(I(T_{\mathcal Q}+|X|T_f+Q))$.
    \item The curator output size is $O(I(L_{\mathcal Q}+Q))$.
    \item The client runs in time $O(I(T_g+S))$.
\end{CompactItemize}
\end{lemma}

For utility, we start with bounding the simpler case where $\mathcal Q$ contains just a single pair of functions.

\begin{lemma}[single pair utility]\label{lmm:utility_smallq}
    Suppose $\mathcal Q$ is supported on a single pair $(f,g)$, and the mechanism is run with $I=J=\Theta(\log(1/\eta))$.
    For every $y\in\R^d$, with probability $1-\eta$, the client output $\hat e(y)$ that satisfies
    \[ |\hat e(y)-KDE_X(y)|\leq\alpha + O\left(\frac{S^{1.5}R^2\log(\tfrac1{\eta})}{\epsilon|X|}\right) . \]
\end{lemma}

The next utility bound is for large or infinite $\mathcal Q$.

\begin{lemma}[utility]\label{lmm:utility_largeq}
    Suppose the mechanism is run with $J=\Theta(\log(1/\eta))$ and $I=\Theta(J/\alpha^2)$. 
    For every $y\in\R^d$, with probability $1-\eta$, the client output $\hat e(y)$ that satisfies
    \[ \left|\hat e(y) - KDE_X(y)\right| < \alpha + O\left(\alpha SR^2 + \frac{S^{1.5}R^2\log(\tfrac1{\eta})}{\alpha\epsilon|X|}\right) . \]
\end{lemma}


\section{DP-KDE for LSQable Kernels}

\subsection{DP-KDE via Random Fourier Features (RFF)}
We recall the construction of RFF for the Gaussian kernel. 
To sample a random feature, one draws $\omega\sim N(0,I_d)$ and $\beta$ uniformly at random over $[0,2\pi)$, and defines the Fourier feature $z_{\omega,\beta}:\R^d\rightarrow\R$ as $z_{\omega,\beta}(x)=\sqrt2\cos(\sqrt2\omega^T x+\beta)$. 
For every $x,y,\in\R^d$ it holds that
\[ 
  e^{-\norm{x-y}_2^2} = \E_{\omega,\beta}[z_{\omega,\beta}(x)\cdot z_{\omega,\beta}(y)] .
\]
This clearly implies an LSQ family $\mathcal Q$, given by sampling $\omega$ and $\beta$ as above and returning the pair $(z_{\omega,\beta},z_{\omega,\beta})$. Since $z_{\omega,\beta}$ takes values in $[-\sqrt2,\sqrt2]$, we obtain,
\begin{proposition}
The Gaussian kernel admits an exact $(1,\sqrt2,1)$-LSQ family.
\end{proposition}

This leads to our first Gaussian DP-KDE mechanism.

\textbf{Proof of \cref{thm:gauhighdim}.}
%
Privacy is guaranteed by \cref{lmm:privacy}. For accuracy we use \cref{lmm:utility_largeq}, plugging $S=1$, $R=\sqrt2$ and $|X|\geq O(\log(1/\eta)/(\epsilon\cdot\alpha^2))$ which holds by the theorem's premise. We get that for every $y\in\R^d$, the client outputs $\hat e(y)$ that with probability $1-\eta$ is off by an additive error of $O(\alpha)$ from the subsampled KDE, and we can scale $\alpha$ by the appropriate constant.
For computational efficiency, note that sampling $(f,g)\sim\mathcal Q$ means sampling $\omega\sim N(0,I_d)$ and $\beta\sim[0,2\pi)$, and takes time $O(d)$; the pair $(f,g)$ can be specified by the $d+1$ machine words $\omega,\beta$; and evaluating $f$ or $g$ on a point in $\R^d$ takes $O(d)$ time. 
Plugging these with $I=O(\log(1/\eta)/\alpha^2)$ (from \cref{lmm:utility_largeq}) into \cref{lmm:efficiency}, we obtain  \cref{thm:gauhighdim}. \qed

\paragraph{Other kernels.} \cite{rahimi2007random} showed that random Fourier features exist all shift-invariant positive definite kernels.
For those kernels, the LSQ framework yields DP-KDE mechanisms with the same error and sample complexity guarantees as the Gaussian kernel in \cref{thm:gauhighdim}. 
However, their computational efficiency may be different, depending on their specific RFF distribution. See \Cref{sec:otherkernels}. 


\subsection{DP-KDE via the Fast Gauss Transform (FGT)}\label{sec:fgt}
We review the Fast Gauss Transform. 
Let all data and query points be contained in a ball $\mathcal B_\Phi$ of radius $\Phi>0$.
Let $\mathcal G$ be the grid with side-length $1$ in $\R^d$ whose nodes are $\mathbb Z^d$. 
Let $\mathcal G_\Phi$ denote the set of $\mathcal G$-grid cells that intersect $\mathcal B_\Phi$.
For every cell $H\in\mathcal G_\Phi$, let $z^H\in\R^d$ denote its center point.

The FGT is based on the Hermite expansion of the Gaussian kernel.
Let $\xi:\R\rightarrow\R$ be defined as $\xi(x)=e^{-x^2}$, and let $\xi^{(r)}$ denote the $r$th derivative of $\xi$ for every $r\geq0$. 
The Hermite function $h_r:\R\rightarrow\R$
is defined as $h_r(x) = (-1)^r\xi^{(r)}(x)$. 
By substituting Taylor series, it can be shown (see \cref{sec:fgtappendix}) that 
for any given $z\in\R^d$, the Gaussian kernel over points in $\R^d$ admits the Hermite expansion,
\[
  e^{-\norm{x-y}_2^2} = \sum_{r_1=1}^\infty \ldots \sum_{r_d=1}^\infty \prod_{j=1}^d\left(x_j-z_j\right)^{r_j}\cdot\frac{1}{r_j!} h_{r_j}\left(y_j-z_j\right) .
\]
Truncating each of the $d$ sums after $\rho=O(\log(1/\alpha))$ terms leads to an additive error of at most $\alpha$. 
%
We can then define the following pair of functions $f,g$ on $\R^d$.
Each of their coordinates is indexed by a pair $H\in\mathcal G_\Phi$ and $r\in\R^d$, where $r$ has coordinates in $\{0,\ldots,\rho\}$, and is set as follows:
\[
  f_{H,r}(x) = \begin{cases}
    \prod_{j=1}^d\left(x_j-z_j^H\right)^{r_j} & \text{if $x\in H$} \\
    0 & \text{otherwise,}
  \end{cases}
\]
\[
  g_{H,r}(y) = \begin{cases}
    \prod_{j=1}^d\frac{1}{r_j!} h_{r_j}\left(y_j-z_j^H\right) & \text{if $\norm{y-z^H}_2^2\leq\rho$} \\
    0 & \text{otherwise.}
  \end{cases}
\]

For usual FGT, one may compute $F(X)=\tfrac1{|X|}\sum_{x\in X}f(x)$ on the dataset $X$, and return $F(X)^Tg(y)$ given a query point $y$.
To our end, we view $(f,g)$ as an LSQ ``family'' with just one pair. In \cref{sec:fgtappendix} we show the following.

\begin{proposition}\label{prp:fgt}
Let $\alpha>0$ be smaller than a sufficiently small constant, and suppose $d=O(\log(1/\alpha))$. The Gaussian kernel over points contained in a Euclidean ball of radius $\Phi$ in $\R^d$ admits an $\alpha$-approximate $(O((1+\frac{\Phi}{\sqrt{d}})(\log(1/\alpha)))^d, O(1)^d, (\log(1/\alpha))^{O(d)})$-LSQ family, supported on a single pair of functions $(f,g)$. 
Furthermore, the evaluation times of $f$ on $x\in\R^d$ and of $g$ on $y\in\R^d$ are both $(\log(1/\alpha))^{O(d)}$. 
\end{proposition}

This yields our second Gaussian DP-KDE mechanism.

\textbf{Proof of \cref{thm:gaulowdim}.}
We may assume w.l.o.g.~that $d=O(\log(1/\alpha))$, since otherwise  \cref{thm:gauhighdim} subsumes \cref{thm:gaulowdim}.
Privacy follows from \cref{lmm:privacy}. For utility we use \cref{lmm:utility_smallq}. By plugging $R,S$ from \cref{prp:fgt}, 
the additive error is, with probability $1-\eta$, at most $\alpha+(\epsilon|X|)^{-1}\log(1/\eta)\cdot(\log(1/\alpha))^{O(d)}$. 
By the lower bound on $|X|$ in the theorem statement, this error is at most $O(\alpha)$, and we can scale $\alpha$ by a constant. For efficiency, note that  having only one pair in $\mathcal Q$ means that $T_{\mathcal Q}=O(1)$ and $L_{\mathcal Q}=0$. Plugging these and $Q,R,S,T_f,T_g$ from \Cref{prp:fgt} into \cref{lmm:efficiency} yields the theorem. \qed

\subsection{DP-KDE via Locality Sensitive Hashing (LSH)}\label{sec:lsh}
In this section we observe that if a kernel is LSHable then it is also LSQable, thereby recovering the results of \cite{coleman2020one} for LSHable kernels (which do not include the Gaussian kernel) within the LSQ framework. We recall the relevant definition of kernel LSHability: 

\begin{definition}\label{def:lsh}
A kernel $k:\R^d\times\R^d\rightarrow[0,1]$ is $\alpha$-approximate LSHable if there is a distribution $\mathcal H$ over hash functions $h:\R^d\rightarrow\{0,1\}^*$, such that for every $x,y\in\R^d$,
\[ \left| k(x,y) - \Pr_{h\sim\mathcal H}[h(x)=h(y)] \right| \leq \alpha . \]
\end{definition}
Suppose the hash functions in $\mathcal H$ map points in $\R^d$ to one of $B$ hash buckets. 
For every $h\in\mathcal H$, let $f_h:\R^d\rightarrow\{0,1\}^B$ map $x$ to the indicator vector of $h(x)$. 
To get an LSQ family $\mathcal Q$ from $\mathcal H$, we may sample $h\sim\mathcal H$ and return the pair $(f_h,f_h)$. For all $x,y\in\R^d$ we thus get $f_h(x)^Tf_h(y)=1$ if $h(x)=h(y)$ and $f_h(x)^Tf_h(y)=0$ if $h(x)\neq h(y)$, hence $\E_{(f,g)\sim\mathcal Q}[f(x)^Tg(y)]=\Pr_{h\sim\mathcal H}[h(x)=h(y)]$. Therefore,
\begin{proposition}\label{prp:lsh2lsq_b}
If $k$ is $\alpha$-approximate LSHable with $B$ hash buckets, then $k$ is $\alpha$-approximate $(B,1,1)$-LSQable.
\end{proposition}

This does not immediately lead to efficient DP-KDE, since $B$ can be very large. For example, all known LSHability results for the Laplacian kernel use $B=\exp(d)$ \cite{rahimi2007random,andoni2009nearest,backurs2019space}. This issue does not typically interfere with non-private applications of LSH, due to sparsity (only one bucket is non-empty per point), but in the DP case, this would disclose information about which buckets are empty. Our LSQ mechanism adds noise to each bucket, which would take time proportional to $B$.
Nonetheless, this can be remedied by standard universal hashing; see \cref{sec:lshappendix}.

\begin{proposition}\label{prp:lsh2lsq}
If $k$ is $\alpha$-approximate LSHable, then $k$ is $2\alpha$-approximate $(\lceil1/\alpha\rceil,1,1)$-LSQable.
\end{proposition}

Together with \cref{lmm:privacy,lmm:efficiency,lmm:utility_smallq,lmm:utility_largeq}, this recovers the DP-KDE results for LSHable kernels within the LSQ framework. As a concrete example, we re-derive a result of \cite{coleman2020one} for the Laplacian kernel. 
\begin{theorem}
\label{thm:laplaciankde}
There is an $\epsilon$-DP function release mechanism for $(\alpha,\eta)$-approximation of Laplacian KDE on datasets in $\R^d$ of size $n\geq O(\log(1/\eta)/(\epsilon\alpha^2))$, and:
\begin{CompactItemize} 
  \item The curator runs in time $O(nd\log(1/\eta)/\alpha^2)$.
  \item The output size is $O(d\log(1/\eta)/\alpha^2)$.
  \item The client runs in time $O(d\log(1/\eta)/\alpha^2)$. 
\end{CompactItemize}
\end{theorem}
\begin{proof}
The Laplacian kernel is LSHable, hence by \Cref{prp:lsh2lsq}, it is $2\alpha$-approximate $(\lceil1/\alpha\rceil,1,1)$-LSQable. 
By \Cref{lmm:privacy,lmm:utility_largeq}, this implies an $\epsilon$-DP mechanism with $(\alpha,\eta)$-approximation for Laplacian KDE. 
Furthermore, the Laplacian kernel LSH families from \cite{rahimi2007random,andoni2009nearest,backurs2019space} 
have $O(d)$ evaluation time, hashing time and description size.
Viewed as LSQ families, they satisfy $T_{\mathcal Q},T_f,T_g,L_{\mathcal Q}=O(d)$ in the notation of \Cref{lmm:efficiency}, which yields the theorem.
\end{proof}


The Laplacian kernel DP-KDE bounds in \Cref{thm:laplaciankde} are the same as those of the Gaussian kernel in \Cref{thm:gauhighdim}. We also remark that the Laplacian kernel admits an efficient RFF distribution, different than its LSH families. Thus, we can also instantiate the LSQ-RFF mechanism for it. This would lead to an alternative proof of \Cref{thm:laplaciankde}, yielding the same asymptotic bounds via a different DP-KDE mechanism; see \Cref{sec:otherkernels_rff}.

See \Cref{sec:lshable} for an overview of other LSHable kernels. 


\section{Experiments}\label{sec:experiments}
We evaluate our mechanisms on public benchmark datasets in both the high- and low-dimensional regimes. For compatibility, we select datasets often used in prior work on density estimation and clustering:
\begin{CompactItemize}
  \item Covertype: forest cover types ($n=581{,}012$, $d=55$) \cite{blackard1999comparative}
  \item GloVe: word embeddings ($n=1{,}000{,}000$, $d=100$) \cite{pennington2014glove}
  \item Diabetes: age and days in hospital ($n=101{,}766$, $d=2$) \cite{strack2014impact}
  \item NYC Taxi: longitude and latitude ($n=100{,}000$, $d=2$) \cite{new-york-city-taxi-fare-prediction}
\end{CompactItemize}

\begin{figure*}[ht]
\centering
\begin{minipage}[t]{\columnwidth}
\begin{minipage}[b]{\columnwidth}
\includegraphics[width=0.88\linewidth]{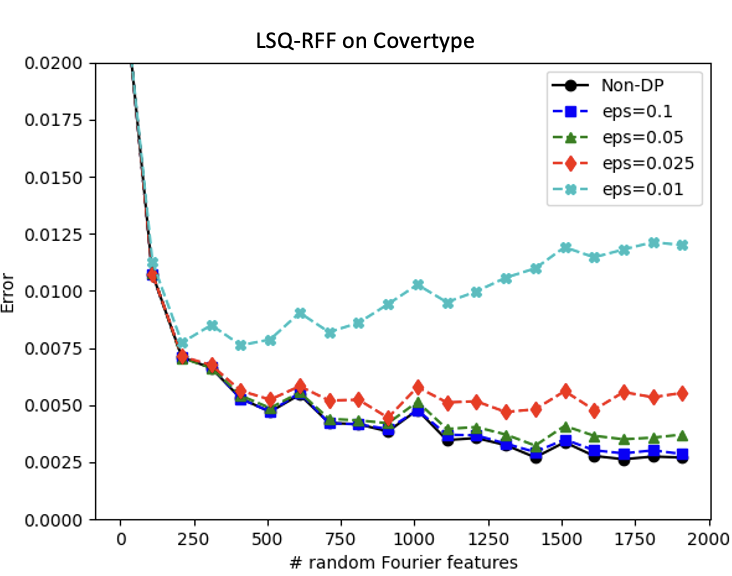}\par
\end{minipage}
\end{minipage}\hfill
\begin{minipage}[t]{\columnwidth}
\begin{minipage}[b]{\columnwidth}
\includegraphics[width=0.88\linewidth]{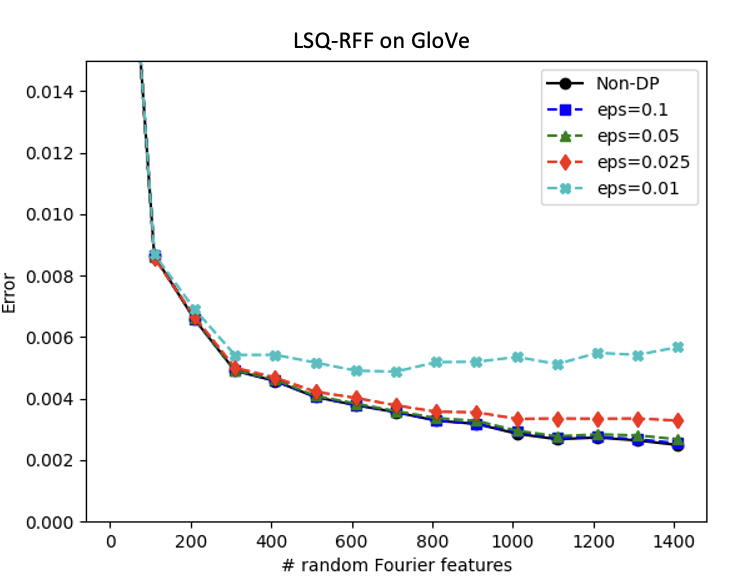}
\end{minipage}
\end{minipage}
\begin{minipage}[t]{\columnwidth}
\begin{minipage}[b]{\columnwidth}
\includegraphics[width=0.88\linewidth]{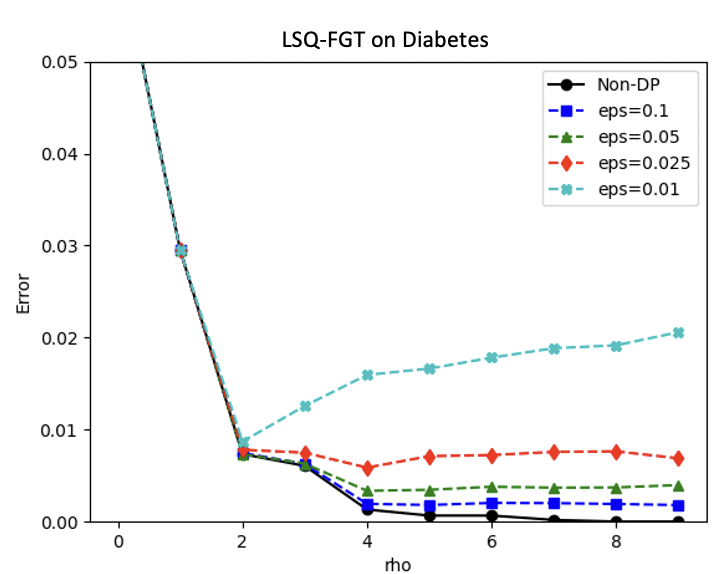}\par
\end{minipage}
\end{minipage}\hfill
\begin{minipage}[t]{\columnwidth}
\begin{minipage}[b]{\columnwidth}
\includegraphics[width=0.88\linewidth]{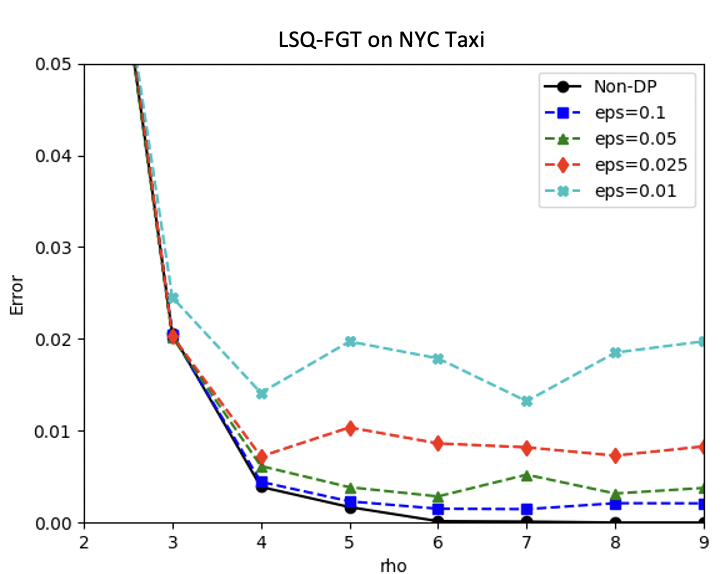}
\end{minipage}
\end{minipage}%
\vspace{-10pt}
\caption{Error vs.~computational budget}
\label{fig:errdiv}
\end{figure*}



Query points are chosen at random from each dataset and are held out from it. 
The reported experimental results are averaged over $100$ queries and $10$ trials with independent random seeds. 
Our code is available online.\footnote{\url{https://github.com/talwagner/lsq}} 
\Cref{sec:experiments_appendix} includes additional details on the implementation of our mechanisms, additional experiments, and more details on our experimental framework and bandwidth selections. 

\subsection{Parameter Selection}\label{sec:error_progression}

In the first experiment, we measure the KDE approximation error of our mechanisms as we increase the parameter that governs their computational budget---the number of Fourier features in RFF, and $\rho$ in FGT. 
\Cref{fig:errdiv} displays the results for several values of $\epsilon$, as well as for a non-private variant of each mechanism, that elides the Laplace noise addition step in \Cref{alg:main} (i.e., it sets $\widetilde F_i=F_i$).

The results highlight a key difference between the DP and non-DP variants: while the error of the non-DP variants converges to zero as we increase their computational budget, the error of the DP mechanisms begins to diverge at a certain point, which corresponds to a smaller parameter setting for smaller $\epsilon$.\footnote{Convergence to zero error is impossible for DP mechanisms due to the sample complexity limitation: for a given dataset size $n$ and $\epsilon>0$, the error $\alpha$ must be large enough to render $n\geq\mathrm{sc}(M)$.} 
This behavior stems from the interplay between non-private approximation and privacy-preserving noise: as we increase the computational budget, the non-private approximation component of the mechanism becomes more accurate, thus disclosing more information about the dataset, that needs to be offset with a larger magnitude of privacy-preserving noise. 
The optimal parameter setting corresponds to the point of balance between the non-private approximation error and the privacy noise error.


For LSQ-RFF, as we increase the number of Fourier features $m$, the error of approximating the Gaussian kernel with (non-private) RFF decays like $1/\sqrt{m}$ by Hoeffding's inequality, while the Laplace noise magnitude grows like $\sqrt{m}/(\epsilon n)$. Hence, the optimal number of Fourier features is $m=\Theta(\epsilon n)$.
Using more features would increase the overall error while having higher computational cost.  

For LSQ-FGT, as we increase $\rho$, the non-private truncation error of the Hermite expansion decays like $\exp(-\rho)$, while the Laplace noise magnitude grows like $\rho^{O(d)}/(\epsilon n)$, hence the optimal setting is $\rho=\Theta(\log(\epsilon n)) - O(d\log\log(\epsilon n))$.

\begin{figure*}[ht]
\centering
\begin{minipage}[t]{\columnwidth}
\begin{minipage}[b]{\columnwidth}
\includegraphics[width=0.88\linewidth]{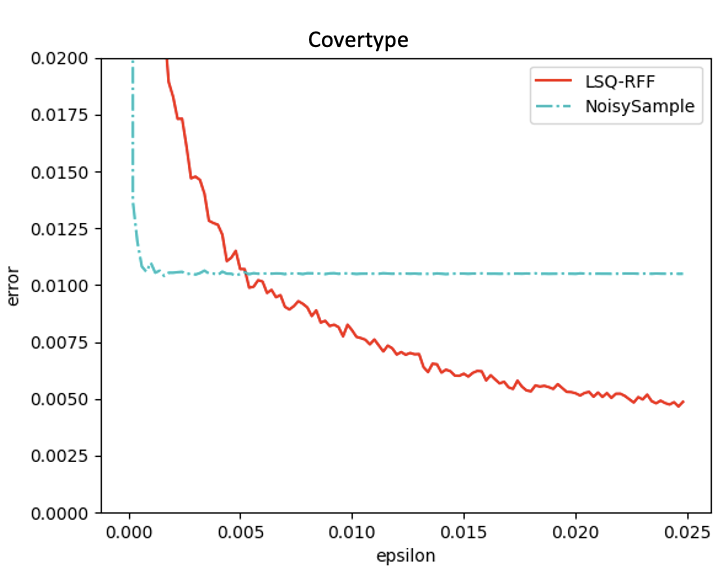}\par
\end{minipage}
\end{minipage}\hfill
\begin{minipage}[t]{\columnwidth}
\begin{minipage}[b]{\columnwidth}
\includegraphics[width=0.88\linewidth]{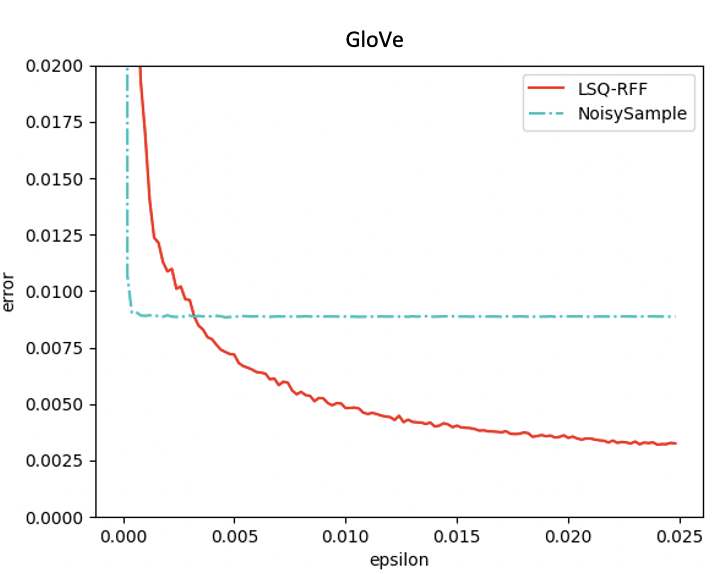}
\end{minipage}
\end{minipage}
\begin{minipage}[t]{\columnwidth}
\begin{minipage}[b]{\columnwidth}
\includegraphics[width=0.88\linewidth]{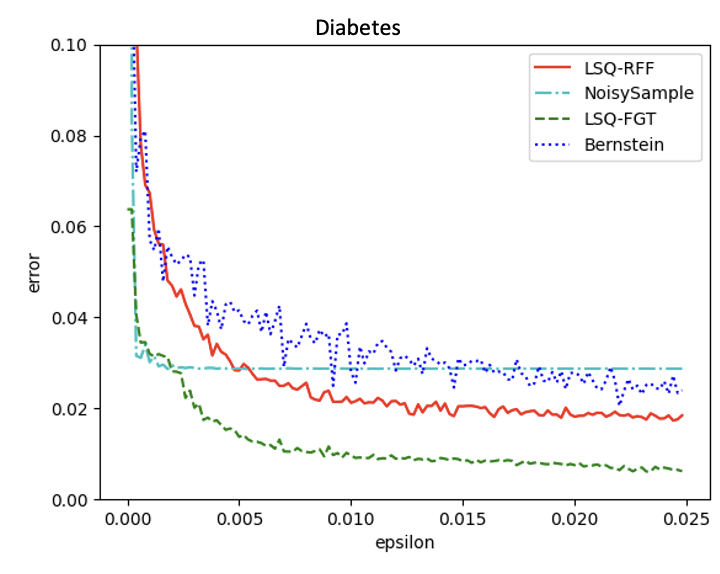}\par
\end{minipage}
\end{minipage}\hfill
\begin{minipage}[t]{\columnwidth}
\begin{minipage}[b]{\columnwidth}
\includegraphics[width=0.88\linewidth]{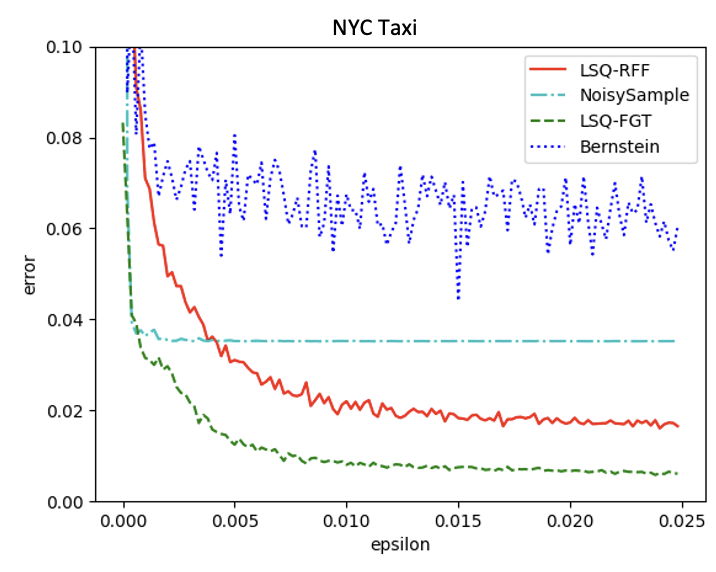}
\end{minipage}
\end{minipage}%
\caption{Error vs.~privacy}
\label{fig:errpriv}
\end{figure*}


\begin{figure*}[ht]
\centering
\begin{minipage}[t]{\columnwidth}
\begin{minipage}[b]{\columnwidth}
\includegraphics[width=0.88\linewidth]{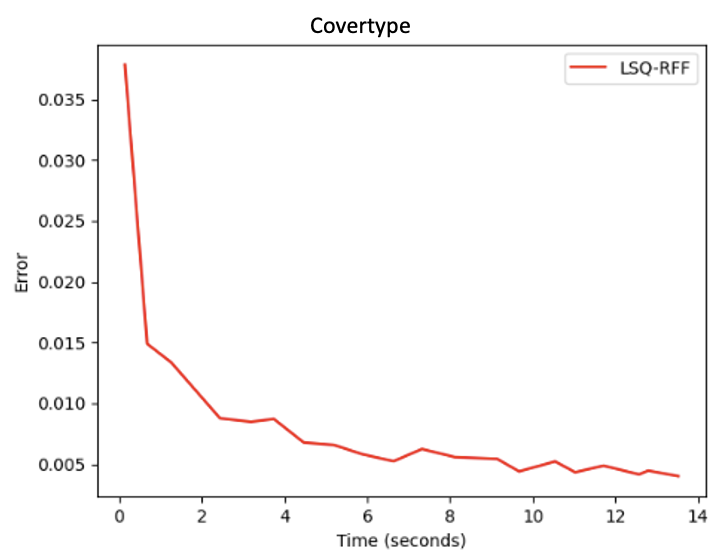}\par
\end{minipage}
\end{minipage}\hfill
\begin{minipage}[t]{\columnwidth}
\begin{minipage}[b]{\columnwidth}
\includegraphics[width=0.88\linewidth]{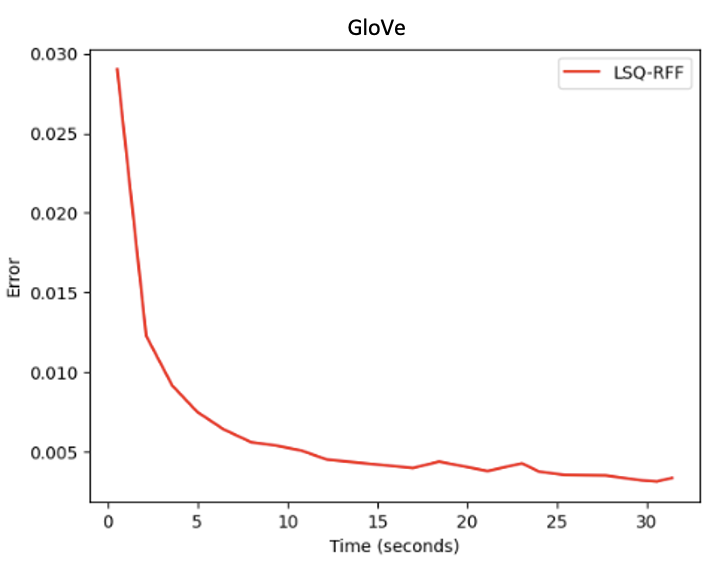}
\end{minipage}
\end{minipage}
\begin{minipage}[t]{\columnwidth}
\begin{minipage}[b]{\columnwidth}
\includegraphics[width=0.88\linewidth]{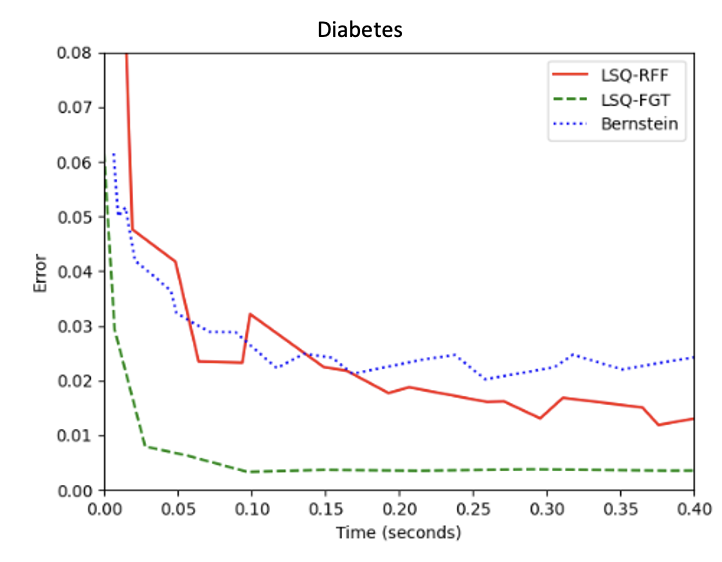}\par
\end{minipage}
\end{minipage}\hfill
\begin{minipage}[t]{\columnwidth}
\begin{minipage}[b]{\columnwidth}
\includegraphics[width=0.88\linewidth]{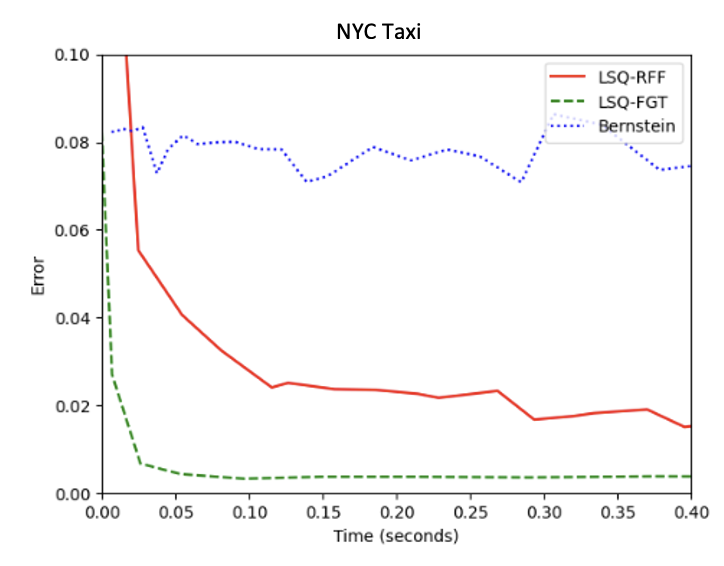}
\end{minipage}
\end{minipage}%
\caption{Error vs. curator running times with $\epsilon=0.05$}
\label{fig:runtime}
\end{figure*}

The upshot is that the parameters should be chosen not only by the available computational budget, but also the desired privacy $\epsilon$ and available dataset size $n$.\footnote{Note that setting the parameters of the mechanism according to the dataset size $n$---e.g., choosing $m\sim\epsilon n$ or $\rho\sim\log(\epsilon n)$---leaks information about $n$ and could affect privacy. This can be easily avoided, for example by using $\tilde n=n+\mathrm{Laplace}(1/\epsilon)$ instead of $n$. It can be easily checked that $\tilde n$ is $\epsilon$-DP and that using it instead of $n$ would only change $m$ or $\rho$ by an additive constant.}

\subsection{Performance}\label{sec:error_privacy}

\textbf{Error vs.~privacy.}
We measure the privacy to error trade-off, with each algorithm evaluated at its optimal setting of parameters for the given value of $\epsilon$. We compare our mechanisms to the following baselines:
\begin{CompactItemize}
  \item NoisySample: A vanilla mechanism that samples $100$ points from the dataset, computes the average of their true KDEs plus a sample from $\mathrm{Laplace}(1/(\epsilon|X|))$, and returns this value as the estimate for any query KDE. The mechanism is $\epsilon$-DP w.r.t.~the non-sampled points. It helps verify that the KDE function is not degenerate and essentially constant (and thus trivial to approximate).
  \item The Bernstein mechanism \cite{alda2017bernstein}, prior state of the art for Gaussian DP-KDE (with pure DP). It has the same error divergence behavior discussed in \Cref{sec:error_progression}, and we evaluate it too at its optimal parameter setting (see \Cref{{sec:expdetails}} for details on Bernstein).
\end{CompactItemize}
The results are in \Cref{fig:errpriv}.
LSQ-FGT and Bernstein are evaluated only on the on the low-dimensional datasets, as they are infeasible for the high-dimensional datasets.

The results show that our LSQ-based mechanisms outperform the baselines by large margins, and procude accurate KDE estimates in desirable privacy regimes. On the low-dimensional datasets, the results corroborate the privacy/error trade-offs predicted by the sample complexity of the mechanisms, as listed in \Cref{tbl:main}.
Note that LSQ-FGT is expected to outperform LSQ-RFF in this regime, due to the near-linear dependence of its sample complexity on $\alpha^{-1}$, compared to the quadratic dependence of LSQ-RFF. 
On the high-dimensional datasets, only LSQ-RFF is feasible.

The performance of Bernstein 
depends on the smoothness of the KDE function, which is determined by the bandwidth $\sigma$ under scaling the data into a unit hypercube (cf.~\Cref{sec:prior}). In particular, its sample complexity depends on $d/\sigma^2$. For NYC Taxi, this quantity is much larger than for Diabetes (cf.~\Cref{sec:experiments_appendix}), accounting for the degraded performance of Bernstein on NYC Taxi compared to Diabetes.

\textbf{Running times.}
We plot the error attained by the mechanisms versus their curator running time. 
\Cref{fig:runtime} reports the results with $\epsilon=0.05$, and \Cref{fig:runtime002} (in the appendix) repeats the experiment with $\epsilon=0.02$. 
Here too, in the high-dimensional regime LSQ-RFF is the only feasible mechanism, while in the low-dimensional regime LSQ-FGT has the best performance.





\section*{Acknowledgements}
We thank Justin Chen, James Cook, Christos Faloutsos, Hakan Ferhatosmanoglu, Supriya Nagesh, Matthew Reimherr, Aaron Roth, Milind Shyani, Doug Terry, Elizabeth Yang, and the anonymous reviewers for helpful feedback, comments and suggestions on this work. 

\bibliography{dpkde}

\begin{thebibliography}{45}
\providecommand{\natexlab}[1]{#1}
\providecommand{\url}[1]{\texttt{#1}}
\expandafter\ifx\csname urlstyle\endcsname\relax
  \providecommand{\doi}[1]{doi: #1}\else
  \providecommand{\doi}{doi: \begingroup \urlstyle{rm}\Url}\fi

\bibitem[Alda \& Rubinstein(2017)Alda and Rubinstein]{alda2017bernstein}
Alda, F. and Rubinstein, B.~I.
\newblock The bernstein mechanism: Function release under differential privacy.
\newblock In \emph{Thirty-First AAAI Conference on Artificial Intelligence},
  2017.

\bibitem[Alman et~al.(2020)Alman, Chu, Schild, and Song]{alman2020algorithms}
Alman, J., Chu, T., Schild, A., and Song, Z.
\newblock Algorithms and hardness for linear algebra on geometric graphs.
\newblock In \emph{2020 IEEE 61st Annual Symposium on Foundations of Computer
  Science (FOCS)}, pp.\  541--552. IEEE, 2020.

\bibitem[Andoni \& Indyk(2009)Andoni and Indyk]{andoni2009nearest}
Andoni, A. and Indyk, P.
\newblock Dimension reduction in kernel spaces from locality-sensitive hashing.
\newblock \emph{Maniscript, also available in Andoni A.,``Nearest neighbor
  search: the old, the new, and the impossible'', PhD thesis, Massachusetts
  Institute of Technology}, 2009.

\bibitem[Aum{\"u}ller et~al.(2021)Aum{\"u}ller, Lebeda, and
  Pagh]{aumuller2021differentially}
Aum{\"u}ller, M., Lebeda, C.~J., and Pagh, R.
\newblock Differentially private sparse vectors with low error, optimal space,
  and fast access.
\newblock In \emph{Proceedings of the 2021 ACM SIGSAC Conference on Computer
  and Communications Security}, pp.\  1223--1236, 2021.

\bibitem[Backurs et~al.(2018)Backurs, Charikar, Indyk, and
  Siminelakis]{backurs2018efficient}
Backurs, A., Charikar, M., Indyk, P., and Siminelakis, P.
\newblock Efficient density evaluation for smooth kernels.
\newblock In \emph{2018 IEEE 59th Annual Symposium on Foundations of Computer
  Science (FOCS)}, pp.\  615--626. IEEE, 2018.

\bibitem[Backurs et~al.(2019)Backurs, Indyk, and Wagner]{backurs2019space}
Backurs, A., Indyk, P., and Wagner, T.
\newblock Space and time efficient kernel density estimation in high
  dimensions.
\newblock \emph{Advances in Neural Information Processing Systems (NeurIPS)},
  2019.

\bibitem[Backurs et~al.(2021)Backurs, Indyk, Musco, and
  Wagner]{backurs2021faster}
Backurs, A., Indyk, P., Musco, C., and Wagner, T.
\newblock Faster kernel matrix algebra via density estimation.
\newblock In \emph{International Conference on Machine Learning (ICML)}, 2021.

\bibitem[Blackard \& Dean(1999)Blackard and Dean]{blackard1999comparative}
Blackard, J.~A. and Dean, D.~J.
\newblock Comparative accuracies of artificial neural networks and discriminant
  analysis in predicting forest cover types from cartographic variables.
\newblock \emph{Computers and electronics in agriculture}, 24\penalty0
  (3):\penalty0 131--151, 1999.
\newblock URL \url{https://archive.ics.uci.edu/ml/datasets/covertype}.

\bibitem[Blocki et~al.(2012)Blocki, Blum, Datta, and
  Sheffet]{blocki2012johnson}
Blocki, J., Blum, A., Datta, A., and Sheffet, O.
\newblock The johnson-lindenstrauss transform itself preserves differential
  privacy.
\newblock In \emph{2012 IEEE 53rd Annual Symposium on Foundations of Computer
  Science (FOCS)}, pp.\  410--419. IEEE, 2012.

\bibitem[Blum et~al.(2013)Blum, Ligett, and Roth]{blum2013learning}
Blum, A., Ligett, K., and Roth, A.
\newblock A learning theory approach to noninteractive database privacy.
\newblock \emph{Journal of the ACM (JACM)}, 60\penalty0 (2):\penalty0 1--25,
  2013.

\bibitem[Carter \& Wegman(1977)Carter and Wegman]{carter1977universal}
Carter, J.~L. and Wegman, M.~N.
\newblock Universal classes of hash functions.
\newblock In \emph{Proceedings of the ninth annual ACM symposium on Theory of
  computing}, pp.\  106--112, 1977.

\bibitem[Charikar \& Siminelakis(2017)Charikar and
  Siminelakis]{charikar2017hashing}
Charikar, M. and Siminelakis, P.
\newblock Hashing-based-estimators for kernel density in high dimensions.
\newblock In \emph{2017 IEEE 58th Annual Symposium on Foundations of Computer
  Science (FOCS)}, pp.\  1032--1043. IEEE, 2017.

\bibitem[Chavez et~al.(2018)Chavez, Sterling, Elliott, V, Sagar, and
  Cukierski]{new-york-city-taxi-fare-prediction}
Chavez, A., Sterling, D., Elliott, J., V, L., Sagar, and Cukierski, W.
\newblock New york city taxi fare prediction, 2018.
\newblock URL
  \url{https://kaggle.com/competitions/new-york-city-taxi-fare-prediction}.

\bibitem[Cherapanamjeri \& Nelson(2020)Cherapanamjeri and
  Nelson]{cherapanamjeri2020adaptive}
Cherapanamjeri, Y. and Nelson, J.
\newblock On adaptive distance estimation.
\newblock \emph{Advances in Neural Information Processing Systems},
  33:\penalty0 11178--11190, 2020.

\bibitem[Chierichetti \& Kumar(2015)Chierichetti and
  Kumar]{chierichetti2015lsh}
Chierichetti, F. and Kumar, R.
\newblock Lsh-preserving functions and their applications.
\newblock \emph{Journal of the ACM (JACM)}, 62\penalty0 (5):\penalty0 1--25,
  2015.

\bibitem[Coleman \& Shrivastava(2020)Coleman and Shrivastava]{coleman2020sub}
Coleman, B. and Shrivastava, A.
\newblock Sub-linear race sketches for approximate kernel density estimation on
  streaming data.
\newblock In \emph{Proceedings of The Web Conference 2020}, pp.\  1739--1749,
  2020.

\bibitem[Coleman \& Shrivastava(2021)Coleman and Shrivastava]{coleman2020one}
Coleman, B. and Shrivastava, A.
\newblock A one-pass distributed and private sketch for kernel sums with
  applications to machine learning at scale.
\newblock In \emph{Proceedings of the ACM SIGSAC Conference on Computer and
  Communications Security (CCS)}, pp.\  3252–3265, 2021.

\bibitem[Cunningham et~al.(2021)Cunningham, Cormode, and
  Ferhatosmanoglu]{cunningham2021privacy}
Cunningham, T., Cormode, G., and Ferhatosmanoglu, H.
\newblock Privacy-preserving synthetic location data in the real world.
\newblock In \emph{17th International Symposium on Spatial and Temporal
  Databases}, pp.\  23--33, 2021.

\bibitem[Dwork et~al.(2006)Dwork, McSherry, Nissim, and
  Smith]{dwork2006calibrating}
Dwork, C., McSherry, F., Nissim, K., and Smith, A.
\newblock Calibrating noise to sensitivity in private data analysis.
\newblock In \emph{Theory of cryptography conference (TCC)}, pp.\  265--284.
  Springer, 2006.

\bibitem[Dwork et~al.(2014)Dwork, Roth, et~al.]{dwork2014algorithmic}
Dwork, C., Roth, A., et~al.
\newblock The algorithmic foundations of differential privacy.
\newblock \emph{Foundations and Trends{\textregistered} in Theoretical Computer
  Science}, 9\penalty0 (3--4):\penalty0 211--407, 2014.

\bibitem[Feldman \& Talwar(2021)Feldman and Talwar]{feldman2021lossless}
Feldman, V. and Talwar, K.
\newblock Lossless compression of efficient private local randomizers.
\newblock In \emph{International Conference on Machine Learning (ICML)}, 2021.

\bibitem[Greengard \& Rokhlin(1987)Greengard and Rokhlin]{greengard1987fast}
Greengard, L. and Rokhlin, V.
\newblock A fast algorithm for particle simulations.
\newblock \emph{Journal of computational physics}, 73\penalty0 (2):\penalty0
  325--348, 1987.

\bibitem[Greengard \& Strain(1991)Greengard and Strain]{greengard1991fast}
Greengard, L. and Strain, J.
\newblock The fast gauss transform.
\newblock \emph{SIAM Journal on Scientific and Statistical Computing},
  12\penalty0 (1):\penalty0 79--94, 1991.

\bibitem[Gupta et~al.(2012)Gupta, Roth, and Ullman]{gupta2012iterative}
Gupta, A., Roth, A., and Ullman, J.
\newblock Iterative constructions and private data release.
\newblock In \emph{Theory of cryptography conference}, pp.\  339--356.
  Springer, 2012.

\bibitem[Hall(2013)]{hall2013new}
Hall, R.
\newblock \emph{New Statistical Applications for Differential Privacy}.
\newblock PhD thesis, Carnegie Mellon University, 2013.

\bibitem[Hall et~al.(2013)Hall, Rinaldo, and Wasserman]{hall2013differential}
Hall, R., Rinaldo, A., and Wasserman, L.
\newblock Differential privacy for functions and functional data.
\newblock \emph{Journal of Machine Learning Research}, 14\penalty0
  (Feb):\penalty0 703--727, 2013.

\bibitem[Hardt \& Rothblum(2010)Hardt and Rothblum]{hardt2010multiplicative}
Hardt, M. and Rothblum, G.~N.
\newblock A multiplicative weights mechanism for privacy-preserving data
  analysis.
\newblock In \emph{2010 IEEE 51st annual symposium on foundations of computer
  science}, pp.\  61--70. IEEE, 2010.

\bibitem[Hofmann et~al.(2008)Hofmann, Sch{\"o}lkopf, and
  Smola]{hofmann2008kernel}
Hofmann, T., Sch{\"o}lkopf, B., and Smola, A.~J.
\newblock Kernel methods in machine learning.
\newblock \emph{The annals of statistics}, 36\penalty0 (3):\penalty0
  1171--1220, 2008.

\bibitem[Huai et~al.(2019)Huai, 0015, Miao, Xu, and Zhang]{huai2019privacy}
Huai, M., 0015, D.~W., Miao, C., Xu, J., and Zhang, A.
\newblock Privacy-aware synthesizing for crowdsourced data.
\newblock In \emph{IJCAI}, pp.\  2542--2548, 2019.

\bibitem[Indyk(2006)]{indyk2006stable}
Indyk, P.
\newblock Stable distributions, pseudorandom generators, embeddings, and data
  stream computation.
\newblock \emph{Journal of the ACM (JACM)}, 53\penalty0 (3):\penalty0 307--323,
  2006.

\bibitem[Indyk \& Motwani(1998)Indyk and Motwani]{indyk1998approximate}
Indyk, P. and Motwani, R.
\newblock Approximate nearest neighbors: towards removing the curse of
  dimensionality.
\newblock In \emph{Proceedings of the thirtieth annual ACM symposium on Theory
  of computing (STOC)}, 1998.

\bibitem[Jaakkola et~al.(1999)Jaakkola, Diekhans, Haussler,
  et~al.]{jaakkola1999using}
Jaakkola, T.~S., Diekhans, M., Haussler, D., et~al.
\newblock Using the fisher kernel method to detect remote protein homologies.
\newblock In \emph{ISMB}, volume~99, pp.\  149--158, 1999.

\bibitem[Johnson \& Schechtman(1982)Johnson and
  Schechtman]{johnson1982embedding}
Johnson, W.~B. and Schechtman, G.
\newblock Embedding l \_p\^{}m into l \_1\^{}n.
\newblock 1982.

\bibitem[Karnin \& Liberty(2019)Karnin and Liberty]{karnin2019discrepancy}
Karnin, Z. and Liberty, E.
\newblock Discrepancy, coresets, and sketches in machine learning.
\newblock In \emph{Conference on Learning Theory}, pp.\  1975--1993. PMLR,
  2019.

\bibitem[Lacoste-Julien et~al.(2015)Lacoste-Julien, Lindsten, and
  Bach]{lacoste2015sequential}
Lacoste-Julien, S., Lindsten, F., and Bach, F.
\newblock Sequential kernel herding: Frank-wolfe optimization for particle
  filtering.
\newblock In \emph{Artificial Intelligence and Statistics}, pp.\  544--552.
  PMLR, 2015.

\bibitem[Lopez-Paz et~al.(2015)Lopez-Paz, Muandet, Sch{\"o}lkopf, and
  Tolstikhin]{lopez2015towards}
Lopez-Paz, D., Muandet, K., Sch{\"o}lkopf, B., and Tolstikhin, I.
\newblock Towards a learning theory of cause-effect inference.
\newblock In \emph{International Conference on Machine Learning}, pp.\
  1452--1461. PMLR, 2015.

\bibitem[Nikolov(2023)]{nikolov2022private}
Nikolov, A.
\newblock Private query release via the johnson-lindenstrauss transform.
\newblock \emph{ACM-SIAM Symposium on Discrete Algorithms (SODA)}, 2023.

\bibitem[Pagh \& Thorup(2022)Pagh and Thorup]{pagh2022improved}
Pagh, R. and Thorup, M.
\newblock Improved utility analysis of private countsketch.
\newblock In \emph{Advances in Neural Information Processing Systems
  (NeurIPS)}, 2022.

\bibitem[Pennington et~al.(2014)Pennington, Socher, and
  Manning]{pennington2014glove}
Pennington, J., Socher, R., and Manning, C.~D.
\newblock Glove: Global vectors for word representation.
\newblock In \emph{Proceedings of the 2014 conference on empirical methods in
  natural language processing (EMNLP)}, pp.\  1532--1543, 2014.
\newblock URL \url{https://nlp.stanford.edu/projects/glove/}.

\bibitem[Phillips \& Tai(2020)Phillips and Tai]{phillips2020near}
Phillips, J.~M. and Tai, W.~M.
\newblock Near-optimal coresets of kernel density estimates.
\newblock \emph{Discrete \& Computational Geometry}, 63\penalty0 (4):\penalty0
  867--887, 2020.

\bibitem[Rahimi \& Recht(2007)Rahimi and Recht]{rahimi2007random}
Rahimi, A. and Recht, B.
\newblock Random features for large-scale kernel machines.
\newblock \emph{Advances in neural information processing systems}, 20, 2007.

\bibitem[Shawe-Taylor et~al.(2004)Shawe-Taylor, Cristianini,
  et~al.]{shawe2004kernel}
Shawe-Taylor, J., Cristianini, N., et~al.
\newblock \emph{Kernel methods for pattern analysis}.
\newblock Cambridge university press, 2004.

\bibitem[Siminelakis et~al.(2019)Siminelakis, Rong, Bailis, Charikar, and
  Levis]{siminelakis2019rehashing}
Siminelakis, P., Rong, K., Bailis, P., Charikar, M., and Levis, P.
\newblock Rehashing kernel evaluation in high dimensions.
\newblock In \emph{International Conference on Machine Learning}, pp.\
  5789--5798. PMLR, 2019.

\bibitem[Strack et~al.(2014)Strack, DeShazo, Gennings, Olmo, Ventura, Cios, and
  Clore]{strack2014impact}
Strack, B., DeShazo, J.~P., Gennings, C., Olmo, J.~L., Ventura, S., Cios,
  K.~J., and Clore, J.~N.
\newblock Impact of hba1c measurement on hospital readmission rates: analysis
  of 70,000 clinical database patient records.
\newblock \emph{BioMed research international}, 2014, 2014.
\newblock URL
  \url{https://archive.ics.uci.edu/ml/datasets/diabetes+130-us+hospitals+for+years+1999-2008}.

\bibitem[Wang et~al.(2016)Wang, Jin, Fan, Zhang, Huang, Zhong, and
  Wang]{wang2016differentially}
Wang, Z., Jin, C., Fan, K., Zhang, J., Huang, J., Zhong, Y., and Wang, L.
\newblock Differentially private data releasing for smooth queries.
\newblock \emph{The Journal of Machine Learning Research}, 17\penalty0
  (1):\penalty0 1779--1820, 2016.

\end{thebibliography}
\bibliographystyle{icml2023}

\newpage
\appendix
\onecolumn
\section{Analysis of the LSQ Mechanism}\label{sec:lsqappendix}

\paragraph{Proof of \cref{lmm:privacy}.}
The $(Q,R,S)$-LSQ property is easily seen to imply that the sensitivity of each $F_i$ in \Cref{alg:main} is $SR/|X|$, and we have $I$ of them, thus the lemma follows from the classical Laplace DP mechanism.

We expand on the details for completeness. 
To this end we recall some DP fundamentals. 
Let $F$ be a function that maps a dataset to $\R^m$.
The \emph{$\ell_1$-sensitivity} of $F$ is defined as $\Delta F=\max_{X,X'}\norm{F(X)-F(X')}_1$, where the maximum is taken over all pairs of neighboring datasets $X,X'$ (the definition of neigboring datasets is given in \Cref{sec:dp}). 
Given a function $F$, a dataset $X$, and $\epsilon>0$, the Laplace mechanism \cite{dwork2006calibrating} releases $F(X)+N$, where $N\in\R^m$ has entries sampled i.i.d.~from $\mathrm{Laplace}(\Delta F/\epsilon)$. This mechanism is $\epsilon$-DP \cite{dwork2006calibrating}.

In \Cref{alg:main} we have $F_i=\frac1{|X|}\sum_{x\in X}f_i(x)$, where $f_i$ is sampled from a $(Q,R,S)$-LSQ family, and thus every $f_i(x)$ has at most $S$ non-zero entries, each of contained in $[-R,R]$. Therefore, $F_i$ has $\ell_1$-sensitivity $RS/|X|$, and the sequence $(F_1,\ldots,F_I)$ has $\ell_1$-sensitivity $IRS/|X|$. The curator in \Cref{alg:main} releases $(\widetilde F_1, \ldots, \widetilde F_I)$, which we observe is but the output of the Laplace mechanism on this function, and is thus $\epsilon$-DP. The curator also releases $(f_i,g_i)_{i=1}^I$, which are sampled obliviously to the dataset and have no effect on differential privacy.
\qed

\paragraph{Proof of \cref{lmm:efficiency}.}
The lemma follows by tracking the steps of the curator and client algorithms. 

Curator running time: In each of $i=1,\ldots,I$ iterations, it samples $(f_i,g_i)\sim\mathcal Q$ in time $T_{\mathcal Q}$, evaluates $f_i$ on every $x\in X$ in total time $|X|T_f$, and adds Laplace noise to each of $Q$ coordinates in total time $O(Q)$. 

Curator output size: For every $i=1,\ldots,I$, it outputs the pair $(f_i,g_i)$ which is described using $L_{\mathcal Q}$ machine words, and the $Q$-dimensional vector $\widetilde F_i$ which occupies $Q$ machine words.

Client running time: For every $i=1,\ldots,I$, it evaluates $g_i(y)$ in time $T_g$, and computes the inner product $\widetilde F_i^Tg_i(y)$, which can be done in time $O(S)$ since $g_i(y)$ has at most $S$ non-zero entries. This takes total time $O(I(T_g+S))$. It then returns the median of $J$ values, each one of whom is the mean of $I'$ values, which takes additional time $O(I'J)=O(I)$.
\qed


\paragraph{Proof of \cref{lmm:utility_smallq}.}
Let $j\in[J]$. By plugging $I'=1$ (since in this lemma we have $I=J$) and $(f_i,g_i)=(f,g)$ (since we have a single pair $(f,g)$) into the definition of $m_j$ in the client algorithm, we have
\begin{equation}\label{eq:mj_smallq}
    m_j = \widetilde F_j^Tg(y) = \frac{1}{|X|}\sum_{x\in X}f(x)^Tg(y) + N_j^Tg(y) , 
\end{equation}
where $N_j=(N_j^{(1)},\ldots,N_j^{(Q)})$ is a random vector whose entries are drawn i.i.d.~from $\mathrm{Laplace}(IRS/(\epsilon|X|))$. 
By properties of the Laplace distribution, each entry $N_j^{(q)}$ has variance $\Var(N_j^{(q)})=2(IRS/(\epsilon|X|))^2$. Since $\Var(N_j^Tg_j(y))=\sum_{q=1}^Qg_j(y)^2\Var(N_j^{(q)})$, and $g_j(y)$ has at most $S$ non-zero entries contained in $[-R,R]$, we have $\Var(N_j^Tg_j(y))\leq SR^2\cdot 2(IRS/(\epsilon|X|))^2$. Thus by Chebyshev's inequality  (recalling that $I=\Theta(\log(1/\eta))$),
\[
  \Pr\left[\left|N_j^Tg_j(y)\right| > \frac{O(1)\cdot S^{1.5}R^2\log(1/\eta)}{\epsilon|X|}\right] < \frac16 . 
\]
Now by the Chernoff inequality, the median of $J=\Theta(\log(1/\eta))$ independent copies $N_1^Tg(y),\ldots,N_J^Tg(y)$ satisfies
\begin{equation}\label{eq:utility_smallq_median}
  \Pr\left[\left|\mathrm{median}(N_1^Tg(y),\ldots,N_J^Tg(y))\right|\geq\frac{O(1)\cdot S^{1.5}R^2\log(1/\eta)}{\epsilon|X|}\right] < \eta .\end{equation}

The client output in \Cref{alg:main} equals $\hat e(y) = \mathrm{median}(m_1,\ldots,m_J)$. 
By noting in \cref{eq:mj_smallq} that the term $\frac{1}{|X|}\sum_{x\in X}f(x)^Tg(y)$ does not depend on $j$, we have 
\[ \hat e(y) = \mathrm{median}(m_1,\ldots,m_J) =  \frac{1}{|X|}\sum_{x\in X}f(x)^Tg(y) + \mathrm{median}(N_1^Tg(y),\ldots,N_J^Tg(y)) . \]
The $\alpha$-approximate LSQ property, in the case of $\mathcal Q$ supported on a single pair, guarantees that $|k(x,y)-f(x)^Tg(y)|\leq\alpha$. Thus,
\begin{align*}
        |\hat e(y)-KDE_X(y)| &= \left| \frac{1}{|X|}\sum_{x\in X}f(x)^Tg(y) + \mathrm{median}(N_1^Tg(y),\ldots,N_J^Tg(y))-KDE_X(y)\right| \\
        &\leq \left| \frac{1}{|X|}\sum_{x\in X}f(x)^Tg(y)-KDE_X(y)\right | + |\mathrm{median}(N_1^Tg(y),\ldots,N_J^Tg(y))| \\
        &\leq \frac{1}{|X|}\sum_{x\in X}\left|f(x)^Tg(y)-k(x,y)\right| + |\mathrm{median}(N_1^Tg(y),\ldots,N_J^Tg(y))| \\
        &\leq \alpha + |\mathrm{median}(N_1^Tg(y),\ldots,N_J^Tg(y))| \\
        &\leq \alpha + \frac{O(1)\cdot S^{1.5}R^2\log(1/\eta)}{\epsilon|X|},
\end{align*}
where the final inequality holds with probability $1-\eta$ by \cref{eq:utility_smallq_median}, as was to be shown. \qed

\paragraph{Proof of \cref{lmm:utility_largeq}.}
For every $i=1,\ldots,I$ we have
\[
    \widetilde F_i^Tg_i(y) 
    = F_i^Tg_i(y) + N_i^Tg_i(y) ,
\]
where $N_i$ is a random vector whose entries are drawn i.i.d.~from $\mathrm{Laplace}(IRS/(\epsilon|X|))$. Therefore, for every $j=1,\ldots,J$,
\begin{equation}\label{eq:util_mj}
    m_j = \frac1{I'}\sum_{i=I'(j-1)+1}^{I'j}\widetilde F_i^Tg_i(y) 
    = \frac1{I'}\sum_{i=I'(j-1)+1}^{I'j}F_i(x)^Tg_i(y) + \frac1{I'}\sum_{i=I'(j-1)+1}^{I'j}N_i^Tg_i(y),
\end{equation}
where we recall that $I'=\lfloor I/J \rfloor = \Theta(1/\alpha^2)$. 
We handle the two sums in turn. 

For the first sum, consider a random choice of $(f_i,g_i)\sim\mathcal Q$, and recall that $F_i(x)=\frac1{|X|}\sum_{x\in X}f_i(x)$. 
By the $(Q,R,S)$-LSQ property of $\mathcal Q$, for every $x,y$ it holds that both $f_i(x)$ and $g_i(y)$ have at most $S$ non-zero entries of magnitude at most $R$, hence $|f_i
(x)^Tg_j(y)|\leq SR^2$. Therefore,
\[ \left|F_i(x)^Tg_i(y)\right| \leq \frac1{|X|}\sum_{x\in X}\left|f_i(x)^Tg_i(y)\right| \leq SR^2 . \]
This holds for every supported pair $(f_i,g_i)$. Consequently, Hoeffding's concentration inequality ensures that averaging $I'=\Theta(1/\alpha^2)$ independent copies of $F_i(x)^Tg_i(y)$ yields
\[
\Pr\left[\left|\frac1{I'}\sum_{i=I'(j-1)+1}^{I'j}F_i(x)^Tg_i(y) - \E\left[F_i(x)^Tg_i(y)\right]\right| > O(1)\cdot\alpha SR^2 \right] < \frac16 .
\]
Moreover, the expectation $\E\left[F_i(x)^Tg_i(y)\right]$ satisfies
\begin{align*}
  \left|\E\left[F_i(x)^Tg_i(y)\right] - KDE_X(y) \right| &= \left|\E\left[ \frac{1}{|X|}\sum_{x\in X}f_i(x)^Tg_i(y)\right] - \frac{1}{|X|}\sum_{x\in X}k(x,y) \right|  \\
  &\leq \frac{1}{|X|}\sum_{x\in X}\left|\E\left[f_i(x)^Tg_i(y)\right]-k(x,y)\right|  \\
  &\leq \alpha, 
\end{align*}
where the final inequality is an application of the $\alpha$-approximate LSQ property of $\mathcal Q$, i.e., $\left|\E\left[f_i(x)^Tg_i(y)\right]-k(x,y)\right|\leq\alpha$.
Combining these, we get
\begin{equation}\label{eq:util_term1}
\Pr\left[\left|\frac1{I'}\sum_{i=I'(j-1)+1}^{I'j}F_i(x)^Tg_i(y) - KDE_X(y)\right| > \alpha + O(1)\cdot\alpha SR^2 \right] < \frac16 .
\end{equation}

For the second sum in \cref{eq:util_mj}, recall that in the above proof of \cref{lmm:utility_smallq} it was shown that $\Var(N_i^Tg_i(y))\leq SR^2\cdot 2(IRS/(\epsilon|X|))^2$ for every $i$. Averaging over $I'=\Theta(1/\alpha^2)$ independent copies scales the variance down by $1/|I'|$, ensuring it is at most $O(1)\cdot\alpha^2 SR^2\cdot (IRS/(\epsilon|X|))^2$. Plugging $I=\Theta(\log(1/\eta)/\alpha^2)$, we have by Chebyshev's inequality,
\begin{equation}\label{eq:util_term2}
\Pr\left[\left|\frac1{I'}\sum_{i=I'(j-1)+1}^{I'j}N_i^Tg_i(y)\right| > \frac{O(1)\cdot S^{1.5}R^2\log(1/\eta)}{\alpha\epsilon|X|}\right] < \frac16 . 
\end{equation}
Taking a union bound over \cref{eq:util_term1,eq:util_term2} and plugging both into \cref{eq:util_mj}, we get
\[ \Pr\left[\left|m_j - KDE_X(y)\right| > \alpha + O(1)\cdot\left(\alpha SR^2 + \frac{S^{1.5}R^2\log(1/\eta)}{\alpha\epsilon|X|}\right)\right] < \frac13 . \]
The client output is $\hat e(y) = \mathrm{median}(m_1,\ldots,m_J)$. Since $J=\Theta(\log(1/\eta))$, we get by Chernoff's inequality,
\[
 \Pr\left[\left|\hat e(y) - KDE_X(y)\right| < \alpha + O(1)\cdot\left(\alpha SR^2 + \frac{S^{1.5}R^2\log(1/\eta)}{\alpha\epsilon|X|}\right)\right] \geq 1-\eta ,
\]
as desired.
\qed

\section{Additional Omitted Analysis}

\subsection{Fast Gauss Transform (\Cref{sec:fgt})}\label{sec:fgtappendix}
For context, we start by deriving the Hermite expansion of the Gaussian kernel. Let $x,y,z\in\R^d$. We may write,
\begin{align*}
  e^{-\norm{y-x}_2^2} &= \prod_{j=1}^d e^{-(y_j-x_j)^2} \\
  &= \prod_{j=1}^d \xi(y_j-x_j) \\
  &= \prod_{j=1}^d\left(\sum_{r_j=1}^\infty\frac{(z_j-x_j)^{r_j}}{r_j!}\cdot\xi^{(r_j)}(y_j-z_j)\right) \\
  &= \prod_{j=1}^d\left(\sum_{r_j=1}^\infty\frac{(x_j-z_j)^{r_j}}{r_j!}\cdot h_{r_j}(y_j-z_j)\right) \\
  &= \sum_{r_1=1}^\infty\ldots\sum_{r_d=1}^\infty\prod_{j=1}^d\frac{(x_j-z_j)^{r_j}}{r_j!}\cdot h_{r_j}(y_j-z_j) ,
\end{align*}
where we recall from \Cref{sec:fgt} that $\xi$ denotes the univariate function $\xi(\gamma)=e^{-\gamma^2}$, that $\xi^{(r)}$ denotes its $r$th derivative, and that $h_{r}$ denotes the Hermite function of order $r$. 
With this notation, the third equality above is by replacing each $\xi(y_j-x_j)$ with its Taylor expansion about $y_j-z_j$. The fourth equality is by recalling that $h_r=(-1)^r\xi^{(r)}$. The fifth equality is by rewriting the product of sums as the sum of products. 

\cite{greengard1991fast} show that truncating each of the $d$ sums after $\rho=O(\log(1/\alpha))$ terms leads to an additive error of at most $\alpha$. Thus,
\begin{equation}\label{eq:hermite}
    \forall z\in\R^d, \;\;\;\; \left| e^{-\frac12\norm{x-y}_2^2} - \sum_{r_1=1}^\rho \ldots \sum_{r_d=1}^\rho \prod_{j=1}^d(x_j-z_j)^{r_j}\cdot\frac{1}{r_j!}\cdot h_{r_j}(y_j-z_j) \right| \leq \alpha .
\end{equation}

We now prove \cref{prp:fgt} with the pair of functions $f,g$ as defined in \Cref{sec:fgt}. 

\paragraph{Proof of \cref{prp:fgt}.}
The LSQ family is supported on the single pair of functions $(f,g)$. 
Note that by the premise $d=O(\log(1/\alpha))$, we may choose $\rho$ that satisfies $\rho\geq d$. 

We start by showing the $\alpha$-approximate LSQ property, which here means that $|f(x)^Tg(y)-e^{-\norm{x-y}_2^2}|\leq\alpha$ for every $x,y\in\R^d$. 
Let $x,y\in\R^d$.
Let $H_x$ be the grid cell that contains $x$. 
Recall that $z^{H_x}\in\R^d$ denotes its center point. 
Note that $f(x)$ is non-zero only in those entries $f_{H,r}(x)$ for which $H=H^x$.
We consider two cases:
\begin{itemize}
  \item If $\norm{y-z^{H_x}}_2\leq\sqrt\rho$, then by the definition of $f$ and $g$ we have
\[ f(x)^Tg(y) = \sum_{r_1=1}^\rho\ldots\sum_{r_d=1}^\rho \prod_{j=1}^d(x_j-z_j^{H_x})^{r_j}\cdot\frac{1}{r_j!}\cdot h_{r_j}(y_j-z_j^{H_x}), \]
hence by \cref{eq:hermite}, $|f(x)^Tg(y)-e^{-\norm{x-y}_2^2}|\leq\alpha$. 
  \item If $\norm{y-z^{H_x}}_2>\sqrt\rho$, then $f(x)^Tg(y)=0$, since there are no entries where both $f(x)$ and $g(y)$ are non-zero. Thus, in this case it suffices to show that $e^{-\norm{x-y}_2^2}\leq\alpha$. Recall that $H_x$ is a hypercube with side-length $1$ centered at $z^{H_x}$ and contains $x$, hence $\norm{x-z^{H_x}}\leq\tfrac12\sqrt{d}\leq\tfrac12\sqrt{\rho}$. Therefore, by the triangle inequality,
\[
  \norm{x-y}_2^2 
  \geq (\norm{y-z^{H_x}}_2 - \norm{x-z^{H_x}}_2)^2
  \geq (\sqrt\rho - \tfrac12\sqrt\rho)^2 = \tfrac14\rho, 
\]
and thus $e^{-\norm{x-y}_2^2}\leq e^{-\rho/4}\leq\alpha$, which holds provided we choose $\rho=O(\log(1/\alpha))$ with an appropriate hidden constant.
\end{itemize}
In both cases we have $|f(x)^Tg(y)-e^{-\norm{x-y}_2^2}|\leq\alpha$, so $\alpha$-approximate LSQability holds. 

Next, we bound the parameters $(Q,R,S)$ of this LSQ pair. 
\begin{itemize}
  \item Quantization $Q$: Each coordinate is indexed by a pair $H\in\mathcal G_\Phi$ and $r\in\{0,\ldots,\rho\}^d$. We recall that $\mathcal G_\Phi$ is the set of grid cells with side-length $1$ that  intersect a ball of radius $\Phi$, hence $|\mathcal G_\Phi|=O(1+\frac{\Phi}{\sqrt{d}})^d$ by a standard volume argument. The number of choices for $r$ is $(\rho+1)^d$, thus $Q=O((1+\frac{\Phi}{\sqrt{d}})\cdot\rho)^d$. 
  \item Range $R$: Observe that $f(x)$ is zero in all coordinates $(H,r_1,\ldots,r_d)$ except those where $H$ is the (unique) grid cell $H_x$ that contains $x$. Since $H_x$ has side-length $1$ and its center point is $z^{H_x}$, we have $\forall_j|x_j-z_j^{H_x}|\leq0.5$, and therefore the magnitude of $f(x)$ at each non-zero coordinate can be bounded as $\left|\prod_{j=1}^d\left(x_j-z_j^{H_x}\right)^{r_j}\right|\leq1$. 
  
  For $g(y)$, we use the following bound from \cite{greengard1991fast}, which is a consequence of Cramer's inequality for Hermite functions: for every $r_1,\ldots,r_d$ and $y\in\R^d$,
\[ \left| \prod_{j=1}^d\frac{1}{r_j!}\cdot h_{r_j}\left(y_j\right) \right| \leq e^{-\norm{y}_2^2}\prod_{j=1}^d\frac{1.09\cdot(\sqrt2)^{r_j}}{\sqrt{r_j!}} . \]
It is not hard to verify that the term $1.09\cdot(\sqrt{2})^r/\sqrt{r!}$ is maximized over non-negative integers $r$ at $r=1$ and is bounded by $1.09\cdot\sqrt{2}<1.6$, hence the right-hand size is upper bounded by $1.6^d$.
\item Sparsity $S$: Again, $f(x)$ is non-zero only at coordinates $f_{H,r}(x)$ such that $H=H_x$, of which there are only $(\rho+1)^d$ (the number of choices for $r\in\{0,\ldots,\rho\}^d$).

As for $g(y)$, it is non-zero only in coordinates $g_{H,r}(x)$ where $H$ is one of the grid cells of $\mathcal G$ that satisfies $\norm{y-z^H}_2\leq\sqrt\rho$. Since the grid has side-length $1$, the number of cells at distance at most $\sqrt\rho$ from any given point $y$ is at most $O(1+\sqrt{\rho/d})^d$, again by a standard volume argument. Accounting also for the $(\rho+1)^d$ possible choices for $r$, the number of non-zero coordinates $g_{H,r}(x)$ is at most $O(1+\sqrt{\rho/d})^d\cdot(\rho+1)^d \leq \rho^{O(d)}$.
\end{itemize}

Finally, we bound the evaluation times $T_f$ and $T_g$ of $f$ and $g$ respectively.
\begin{itemize}
  \item $T_f$: Every non-zero entry of $f(x)$ is the product of $d$ terms, which takes $O(d)$ time to compute. As shown above, $f(x)$ has $(\rho+1)^d$ non-zero entries, thus its total evaluation time of is thus $O(d)\cdot(\rho+1)^d$.
  \item $T_g$: Let $i\geq0$ be an integer. The hermite function $h_i(\gamma)$ is equal to $e^{-\gamma^2} P_i(\gamma)$ for every $\gamma\in\R$, where $P_i$ is the (``physicist's'') Hermite polynomial of degree $i$. Fix a grid cell $H\in\mathcal G_\Phi$. Since $P_i(\gamma)$ and thus $h_i(\gamma)$ can be evaluated in time $O(i)$ for every $i$ and $\gamma$, all values $\{h_i(y_j-z_j^H):i=0,\ldots,\rho\}$ can be computed in time $O(\rho^2)$. With these at hand, for every $r\in\{0,\ldots,\rho\}^d$ and our fixed $H$ we can compute $g_{H,r}(y)$ in time $O(d)$, by multiplying the appropriate pre-computed values. The total evaluation time for a fixed $H$ is thus $O(\rho^2+d)$. As shown above, the number of cells $H$ whose corresponding entries in $g(y)$ are non-zero is $O(1+\sqrt{\rho/d})^d$, leading to a total computation time of $O(1+\sqrt{\rho/d})^d\cdot(\rho^2+d) \leq \rho^{O(d)}$.
\end{itemize}

Recalling that $d\leq\rho=O(\log(1/\alpha))$, the proof is complete. \qed

\paragraph{Refined LSQ for sharper implementation.}
In \Cref{sec:lsq}, for the purpose of asymptotic analysis, we defined LSQ with a uniform bound $R$ on the range of all coordinates in $f$ and $g$. Nonetheless, the coordinates can have different ranges, as the above proof shows for FGT. While it does not change the asymptotic bounds, it can have practical importance in implementation.

Concretely, let $f,g:\R^d\rightarrow\R^Q$.
Suppose we have $S,R^g,R^f_1,\ldots,R^f_Q\geq0$ such that for every $x,y\in\R^d$:
\begin{itemize}
  \item $g(y)$ has at most $S$ non-zero coordinates;
  \item Each coordinate of $g(y)$ is in $[-R^g,R^g]$;
  \item For $i=1,\ldots,Q$, coordinate $i$ of $f(x)$ is in $[-R^f_i,R^f_i]$.
\end{itemize}
The LSQ mechanism in \Cref{alg:main} adds a sample from $\mathrm{Laplace}((\epsilon|X|)^{-1}IRS)$ to each coordinate, to ensure $\epsilon$-DP via the Laplace mechanism. In the refined form of LSQ stated above, since $f$ has sensitivity $\sum_{i=1}^QR^f_i$, it suffices to add a sample from $\mathrm{Laplace}((\epsilon|X|)^{-1}I\sum_{i=1}^QR^f_i)$ to each coordinate to ensure $\epsilon$-DP. 

In the case of FGT, the above proof of \cref{prp:fgt} shows that if a coordinate of $f_{H,r}$ is indexed by a pair $H\in\mathcal G_\Phi$ and $r\in\{0,\ldots,\rho\}^d$, then $f_{H,r}(x)=0$ if $x\notin H$, and otherwise,
\[ |f_{H,r}(x)|\leq\left|\prod_{j=1}^d(x_j-z_j^H)^{r_j}\right| \leq \prod_{j=1}^d\frac1{2^{r_j}} = \frac{1}{2^{\sum_{j=1}^dr_j}} . \]
Therefore,
\[
  \sum_{i=1}^QR^f_i = \sum_{r_1=0}^\rho\ldots\sum_{r_d=0}^\rho\frac{1}{2^{\sum_{j=1}^dr_j}} = \left(\sum_{r=0}^\rho\frac1{2^r}\right)^d = \left(2\left(1-2^{-(\rho-1)}\right)\right)^d .
\]

Asymptotically, this makes no difference to the analysis: by retracing the proof of \cref{lmm:utility_smallq} with this refined LSQ, we get that the error term $O((\epsilon|X|)^{-1}\log(1/\eta)\cdot S^{1.5}R^2)$ from \cref{lmm:utility_smallq} becomes $O((\epsilon|X|)^{-1}\log(1/\eta)\cdot \sqrt{S}\cdot R^g\cdot\sum_{i=1}^QR^f_i)$. Since $R^g=1.6^d$ and $S=\rho^{O(d)}$ in \cref{prp:fgt}, the resulting error is the same in both cases up to hidden constants. However, in practice, adding noise of magnitude only $\left(2\left(1-2^{-(\rho-1)}\right)\right)^d$ instead of $\rho^{O(d)}$ to each coordinate noticeably improves the empirical performance of FGT, while retaining its theoretical guarantees. 

\subsection{Locality Sensitive Hashing (\Cref{sec:lsh})}\label{sec:lshappendix}
\paragraph{Proof of \cref{prp:lsh2lsq}.}
The proof is by composing a usual pairwise independent hash function over the LSH function. Let $\mathcal U$ be a universal family of hash functions from $\{1,\ldots,B\}$ to $\mathcal\{1,\ldots,B'\}$, where $B$ is the number of buckets in the range of $\mathcal H$, and $B'>0$ is an integer of our choice. We recall that, as per the definition of universal hashing, $\mathcal U$ satisfies $\Pr_{u\sim\mathcal U}[u(b)=u(b')]\leq1/B'$ for every $b,b'$. 

We define an LSQ family $\mathcal Q$ as follows: to sample from it, we draw $h\sim\mathcal H$ and $u\sim\mathcal U$, and for every $x\in\R^d$ we let $f_{h,u}(x)\in\{0,1\}^{B'}$ be the indicator vector for $u(h(x))$. We return $(f_{h,u},f_{h,u})$ as the sampled pair from $\mathcal Q$. 

A union bound over the collision probabilities of $h$ and $u$ yields that for all $x,y\in\R^d$,
\[ \Pr_{u,h}[u(h(x))=u(h(y))] \leq \Pr_{h}[h(x)=h(y)] + \tfrac1{B'} . \]
Consequently, if $\mathcal H$ is an $\alpha$-approximate LSH family for $k$, then $\mathcal Q$ is an $(\alpha+\tfrac1{B'})$-approximate $(B',1,1)$-LSQ family for $k$. \Cref{prp:lsh2lsq} follows by choosing $B'=\lceil1/\alpha\rceil$.
\qed 

\section{Expanded Discussion on Related Work}\label{sec:linearappendix}

\subsection{Generic Linear Queries}\label{sec:gen_lin_queries}
For completeness of the discussion of prior work from \Cref{sec:prior}, we expand on some aspects of SmallDB and PMW for Gaussian DP-KDE in the function release model. 
These mechanisms are designed to answer generic linear queries. 
Let $\mathcal X$ denote the universe in which the elements of the dataset $X$ are contained. 
The goal of a DP linear query is to estimate the quantity $\phi(X):=\frac{1}{|X|}\sum_{x\in X}\phi(x)$, where $\phi:\mathcal X\rightarrow[0,1]$ is a query function chosen by the client.
In the case of KDE, we have $\mathcal X=\R^d$, and each query point $y\in\R^d$ corresponds to the query function $\phi_y(x)=k(x,y)$. 
Since SmallDB and PMW require $\mathcal X$ to be finite, we next discuss discretization.

\textbf{Discretization for Gaussian KDE.}
If all points are assumed to be contained in ball of radius $\Phi\geq1$ in $\R^d$, then for the purpose of approximation of Gaussian KDE (\cref{def:addapx}), one can round every point coordinate to its nearest integer multiple of $\alpha/(4\Phi\sqrt{d})$. Thus, we can without loss of generality assume that $\mathcal X$ contains only those points in the ball that have such coordinates, of which there are $O(\Phi^2/\alpha)^d$ by a standard volume argument. 

To see why this suffices for Gaussian KDE, let $x,y\in\R^d$, and let $\bar x$ be the result of rounding $x$. Then,
\[
 e^{-\norm{y-\bar x}_2^2} = e^{ -\norm{(y-x) - (x-\bar x)}_2^2}
  = e^{-\norm{y-x}_2^2} \cdot e^{-\norm{x-\bar x}_2^2} \cdot e^{2(y-x)^T(x-\bar x)}.
\]
Since $\norm{x-\bar x}_2\leq \alpha/(4\Phi)$,
\[ 1 \geq e^{-\norm{x-\bar x}_2^2} \geq e^{-(\alpha/(4\Phi))^2} \geq e^{-\alpha} , \]
and, by Cauchy-Schwartz and the fact that $\norm{y-x}_2\leq2\Phi$,
\[
  |2(y-x)^T(x-\bar x)| \leq 2\norm{y-x}_2 \norm{x-\bar x}_2 \leq 2\cdot 2\Phi\cdot\frac{\alpha}{4\Phi} = \alpha,
\]
which implies 
\[
  e^{-\alpha}\leq e^{2(y-x)^T(x-\bar x)} \leq e^\alpha . 
\]
Noting that $1\leq e^\alpha \leq 1+2\alpha$ and $1-\alpha\leq e^{-\alpha}\leq 1$ for all $\alpha\in(0,1)$, we plug these back above and get,
\[ (1-\alpha)^2 e^{-\norm{y-x}_2^2} \leq e^{-\norm{y-\bar x}_2^2} \leq (1+2\alpha)e^{-\norm{y-x}_2^2}, \]
thus $|e^{-\norm{y-\bar x}_2^2}-e^{-\norm{y-x}_2^2}|\leq 2\alpha\cdot e^{-\norm{y-x}_2^2}\leq2\alpha$. 
Therefore rounding up to this precision introduces an additive error of only $O(\alpha)$ to every kernel evaluation and hence to every KDE evaluation, and we can scale $\alpha$ down by an appropriate constant.

\textbf{SmallDB.}
The mechanism works as follows: Let $X$ be the curator dataset, and let $Q$ be a set of client queries. Suppose we know of $s(\alpha,Q)\geq0$ such that there exists a dataset $Z$ of size $s(\alpha,Q)$ that satisfies $|\phi(Z)-\phi(X)|\leq\alpha$ for all $\phi\in Q$ simultaneously. SmallDB selects a dataset $\widetilde Z$ of size $s(\alpha,Q)$ using the DP exponential mechanism, and, in the query release model, releases the answers $\{\phi(\widetilde Z):\phi\in Q\}$ to the client queries $Q$. 
When the goal is to release $\epsilon$-DP accurate answers to all queries in $Q$ simultaneously with constant probability (say $0.9$), SmallDB has sample complexity $O(s(\alpha,Q)\cdot\log(|\mathcal X|)/(\epsilon\alpha))$.

(The exponential mechanism entails iterating over all possible datasets of size $s(\alpha,Q)$---that is, all $|\mathcal X|^{s(\alpha,Q)}$ subsets of $\mathcal X$ of that size---and computing their utility with respect to $Q$, which leads to the inefficient running time of SmallDB.)

By standard concentration (Hoeffding's inequality), it is well-known that $s(\alpha,Q)=O(\log(|Q|)/\alpha^2)$ for every $Q$ and $\alpha$, yielding a sample complexity of $O(\log(|Q|)\log(|\mathcal X|))/(\epsilon\alpha^3))$ for generic linear queries. 
In the transformation from query release to function release for Gaussian DP-KDE, we set $Q=\mathcal X$. By discretization we have $|\mathcal X|=O(\Phi^2/\alpha)^d$, hence the above sample complexity becomes $O(d^2\log^2(\Phi/\alpha)/(\epsilon\alpha^3))$.  However, it can be improved somewhat further, due to the existence of coresets for Gaussian KDE. An $\alpha$-coreset for $X$ is a dataset $Z$ such that $|KDE_X(y)-KDE_Z(y)|\leq\alpha$ for all $y\in\R^d$ simultaneously. It is known that every dataset has an $\alpha$-coreset for Gaussian KDE of size
\[ C_{d,\alpha} = O(\min\{\alpha^{-1}\sqrt{d\log(1/\alpha)},\alpha^{-2}\}) , \] 
see \cite{lopez2015towards,lacoste2015sequential,phillips2020near,karnin2019discrepancy}.
Therefore, $C_{d,\alpha}$ is an upper bound on $s(\alpha,Q)$ for every $Q$ and $\alpha$. This yields the SmallDB sample complexity bound listed in \Cref{tbl:main}.

The curator running time, which as mentioned above depends on enumerating over all datasets of size $s(\alpha,Q)$, is similarly improved. 
The curator output is the synthetic dataset $\widetilde Z$ released by the exponential mechanism, and it contains $C_{d,\alpha}$ points in $\R^d$, hence its size is $O(d\cdot C_{d,\alpha})$ words. The client can estimate $KDE_X(y)$ on this output by computing $KDE_{\widetilde Z}(y)$, which takes time $O(d\cdot C_{d,\alpha})$.

\textbf{PMW.} The mechanism has sample complexity $\tilde O(\log(|Q|)\log(|\mathcal X|))/(\epsilon\alpha^3))$ for generic linear queries. It is similar to that of SmallDB up to log factors, but stems from a different analysis (that we do not revisit here) which is not immediately improved by the existence of coresets. In the DP-KDE function release case we have, as above, $|Q|=|\mathcal X|=O(\Phi^2/\alpha)^d$, leading to the sample complexity listed in \Cref{tbl:main}. 

(We remark that PMW, unlike SmallDB, allows for adaptive queries in the query release model. Since we transform both mechanisms to the function release model for DP-KDE, this distinction between them does not apply in our setting.)

Like SmallDB, the output of PMW (in the function release model) is a synthetic private dataset $\widetilde Z$ on which the KDE of every query point can be directly evaluated. Initially $\widetilde Z$ can be as large as $\mathcal X$, but it too can be replaced by a coreset of itself, increasing the additive error of every query by at most $\alpha$. The coresets bounds listed above are constructive (in particular, a uniformly random sample of $O(1/\alpha^2)$ from $\widetilde Z$ yields an $\alpha$-coreset for it with constant probability \cite{lopez2015towards}), and since the released coreset would be computed from $\tilde Z$ which is already $\epsilon$-DP, the coreset too would be $\epsilon$-DP by immunity of differential privacy to post-processing. Consequently, like SmallDB, the curator output size and client running time of PMW are both $O(d\cdot C_{d,\alpha})$.

\paragraph{From query release to function release: uniform convergence and running time.}
As alluded to above, a na\"ive way to transform a query release mechanism into a function release mechanism is to invoke it with all possible queries, of which (by the above discretization argument) we have $O(\Phi^2/\alpha)^d$. This was used above to determine the sample complexity bounds for SmallDB and PMW. In fact, by uniform convergence results from learning theory, invoking these mechanisms with a small random sample of queries (instead of all possible queries) suffices to turn them into function release mechanisms. The reason is that the functions these mechanisms release admit a short description (a small synthetic dataset in the case of SmallDB, or a short transcript that describes the synthetic dataset in the case of PMW), and therefore the released functions can be ``learned'' on a small sample of queries and still generalize (in the learning theory sense) to all queries. We omit further details. This argument does not change the sample complexity of these mechanisms, but it somewhat improves the curator running time (albeit it remains at least exponential in $d$, as listed in \Cref{tbl:main}).

\subsection{Adaptive Queries}\label{app:adaptive}
In \Cref{sec:prior} we mentioned that SmallDB and PMW, when used in the function release model, have the property that with a fixed probability of say $0.9$, they release a function\footnote{In the case of SmallDB and PMW, the released function in fact takes the form of a synthetic dataset.} which returns the correct answer up to an additive error of at most $\alpha$ for all queries simultaneously (assuming all points are contained in a ball of radius $\Phi$). This is a stronger guranatee than $(\alpha,\eta)$-approximation. In particular, it allows to use the released function for adaptive queries. 

We can achieve the same stronger guarantee for our mechanisms (and similarly for the Bernstein mechanism), by setting $\eta$ sufficiently small so as to allow for a union bound over all queries (namely, by the above discretization bound, $\eta=\Theta(\alpha/\Phi^2)^d$). We get the following corollaries of \cref{thm:gauhighdim,thm:gaulowdim} respectively. 

\begin{corollary}[high dimensions]\label{thm:gauhighdim_adaptive}
There is an $\epsilon$-DP function release mechanism for Gaussian KDE on datasets in $\R^d$ of size $n\geq O(d\log(\Phi/\alpha)/(\epsilon\alpha^2))$ and that are contained in a ball of radius $\Phi$, such that with probability $0.9$, the released function has additive error at most $\alpha$ on every query simultaneously. Furthermore:
\begin{CompactItemize} 
  \item The curator runs in time $O(nd^2\log(\Phi/\alpha)/\alpha^2)$.
  \item The output size is $O(d^2\log(\Phi/\alpha)/\alpha^2)$.
  \item The client runs in time $O(d^2\log(\Phi/\alpha)/\alpha^2)$. 
\end{CompactItemize}
\end{corollary}

\begin{corollary}[low dimensions]\label{thm:gaulowdim_adaptive}
There is an $\epsilon$-DP function release mechanism for Gaussian KDE on datasets in $\R^d$ of size $n\geq\log(1/\eta) \cdot (\log(1/\alpha))^{O(d)}/(\epsilon\alpha)$ and that are contained in a ball of radius $\Phi$, such that with probability $0.9$, the released function has additive error at most $\alpha$ on every query simultaneously. Furthermore:
\begin{CompactItemize} 
  \item The curator runs in time $(nd+(\frac{\Phi}{\sqrt{d}})^d) \cdot O(\log(1/\alpha))^{O(d)}\cdot d\log(\Phi/\alpha)$.
  \item The output size is $O((1+\frac{\Phi}{\sqrt{d}})(\log(1/\alpha)))^d\cdot d\log(\Phi/\alpha)$.
  \item The client runs in time $(\log(1/\alpha))^{O(d)}\cdot d\log(\Phi/\alpha)$. 
\end{CompactItemize}
\end{corollary}

Note that the dependence on $\Phi$ remains polylogarithmic, and for the first mechanism, the dependence on the dimension remains polynomial.

\subsection{$(\epsilon,\delta)$-DP and Query Release}\label{sec:apxdp}
When $(\epsilon,\delta)$-DP with $\delta>0$ (a.k.a.~approximate DP) is allowed, the most notable prior result on Gaussian DP-KDE is due to \cite{hall2013differential}, which we call the HRW mechanism. Their mechanism is time-efficient in the query release model, albeit not in the function release model. To describe it, we define the query release model as follows. First, the client sends the curator $q$ query points, $y_1,\ldots,y_q\in\R^d$. In response the curator, who holds a dataset $X$, releases a sequence of answers $A=(a_1,\ldots,a_q)$. We require that \emph{(i)} $A$ is differentially private w.r.t.~$X$, and \emph{(ii)} with probability (say) $0.99$, it holds that $\max_{i=1,\ldots,q}|a_i-KDE_X(y_i)|\leq\alpha$.\footnote{This is the \emph{batch} query release model. In the \emph{online} query release model, the client may send the curator additional queries after seeing the answers to previous ones. The results we describe in this section extend to the online variant as well.}

Note that in the query release model, no ``curse of dimensionality'' immediately arises at all: the curator can simply compute the true KDE values of all queries in time $O(dnq)$, and release them after adding appropriate privacy-preserving noise.\footnote{Note that this would not have been possible in the function release model, where the curator has no access to the queries, and no party has to access to both the dataset and the queries simultaneously, thus the true KDE values cannot be computed at all---unless the curator enumerates over all possible queries in advance, before receiving any specific queries from the client.} However, such na\"ive mechanisms lead to an undesirably large sample complexity (or equivalently, undesirably large error $\alpha$), and improving the sample complexity while avoiding exponential dependence on $d$ turns out to be challenging.
This is manifested in the following discussion, whose quantitative results are summarized in \Cref{tbl:queryrelease}.

\begin{table*}
\caption{\small $\epsilon$-DP and $(\epsilon,\delta)$-DP KDE query release mechanisms for the Gaussian kernel, that receive $q$ queries and approximate each KDE up to additive error $\alpha$. 
(*) SmallDB, PMW and LSQ-FGT assume that all points lie in a ball of radius $\Phi$. ($\ddagger$) Recall that $O(\Phi^2/\alpha)^d$ is an upper bound on the number of possible points (cf.~discretization in \Cref{sec:gen_lin_queries}), hence on the number of queries $q$, hence the $d\log(\Phi/\alpha)$ term in the sample complexity of SmallDB is at least $\Omega(\log q)$.}
{\renewcommand{\arraystretch}{2}
\begin{center}
\begin{tabular}{llcccl}
\toprule
\sc Mechanism & \sc Pure DP? & \sc Sample complexity & \sc Runtime dependence on $d$ & \\
\midrule
Laplace & Yes & $O(\frac{q\log q}{\epsilon \cdot \alpha})$ & Linear & \\
SmallDB  & Yes & $O\left(\frac{\min\{\sqrt{d\log(1/\alpha)},1/\alpha\} \cdot d\log(\Phi/\alpha)}{\epsilon \cdot \alpha^2}\right)$ & Exponential &  \scriptsize{(*), ($\ddagger$)} \\
PMW & Yes & $\tilde O\left( \frac{\log (q) \cdot d\log(\Phi/\alpha)}{\epsilon \cdot \alpha^3}\right)$ & Exponential & \scriptsize{(*)} \\
LSQ-RFF &  Yes & $O(\frac{\log q}{\epsilon \cdot \alpha^2})$ & Linear & \\
LSQ-FGT & Yes & $\frac{\log q}{\epsilon \cdot \alpha} \cdot (\log(1/\alpha))^{O(d)}$ & Exponential & \scriptsize{(*)} \\
\midrule
Gaussian & No & $O(\frac{\sqrt{q\log q\log(1/\delta)}}{\epsilon \cdot \alpha})$ & Linear & \\
PMW & No & $\tilde O\left( \frac{\log q \sqrt{d\log(\Phi/\alpha)\log(1/\delta)}}{\epsilon \cdot \alpha^2}\right)$ & Exponential & \scriptsize{(*)} \\
HRW & No & $O(\frac{\sqrt{\log q \log(1/\delta)}}{\epsilon \cdot \alpha})$ & Linear & \\
\bottomrule
\end{tabular}
\end{center}}
\label{tbl:queryrelease}
\end{table*}

\paragraph{Query release with pure DP.}
For context, let us start with DP-KDE in the query release model under pure DP, that is, where the released sequence of answers $A$ must be $\epsilon$-DP w.r.t.~$X$.
As alluded to above, the curator can invoke the vanilla Laplace mechanism: compute the true KDE values of the $q$ queries, and add noise sampled independently from $\mathrm{Laplace}(q/(\epsilon n))$ to each. It is not hard to verify that $A$ has $\ell_1$-sensitivity $q/n$, hence the mechanism is $\epsilon$-DP. The running time is $O(dnq)$. 
The resulting sample complexity is $O(q\log(q)/(\epsilon \alpha))$. 
While the dependence on $d,\alpha,\epsilon$ is desirable, the dependence on $q$ in the sample complexity impedes the usability of this mechanism if the number of queries is large.

Instead of the Laplace mechanism, one could use SmallDB or PMW, whose sample complexity has better dependence on $q$ in some regimes, albeit their running time is (at least) exponential in $d$. 
LSQ-RFF achieves a sample complexity of $O(\log(q)/(\epsilon\alpha^2))$ and running time linear in $d$, subsuming SmallDB and PMW on both counts.\footnote{Of course, LSQ-RFF is specialized for KDE queries, while SmallDB and PMW apply to general linear queries.} Comparing its sample complexity to the Laplace mechanism, the dependence on $q$ is exponentially better, while the dependence on $\alpha$ is quadratically worse.  LSQ-FGT has sample complexity $O(\log(q)\cdot(\log(1/\alpha))^d/(\epsilon\alpha))$ and running time exponential in $d$, improving over the above mentioned results only when $d$ is small.

 
\paragraph{Query release with approximate DP: the HRW mechanism.}
Now suppose approximate DP is allowed---that is, the curator is allowed to release an answer sequence $A$ which is $(\epsilon,\delta)$-DP w.r.t.~X. 
The natural analog of the vanilla Laplace mechanism from the pure DP case is the vanilla Gaussian mechanism (see \cite{dwork2014algorithmic}): the curator computes the true KDE values of all queries, and adds independent Gaussian noise $N(0,2q\log(1.25/\delta)/(\epsilon n)^2)$ to each.
It is not hard to verify that $A$ has $\ell_2$-sensitivity $\sqrt{q}/n$, hence the mechanism is $(\epsilon,\delta)$-DP.
The running time is $O(dnq)$. 
The resulting sample complexity is $O(\sqrt{q\log(q)\cdot\log(1/\delta)}/(\epsilon \alpha))$. 
While the dependence on $q$ is quadratically better than the pure-DP Laplace mechanism, it is still undesirably large. Again, one could use SmallDB or PMW, but they are subsumed by the pure-DP LSQ-RFF mechanism, even when approximate DP is allowed.\footnote{PMW has an $(\epsilon,\delta)$-DP variant with better bounds than its pure-DP variant. SmallDB has no $(\epsilon,\delta)$-DP variant. See \Cref{tbl:queryrelease}.}



\cite{hall2013differential} presented the HRW mechanism, which is $(\epsilon,\delta)$-DP, runs in time $O(dq(n+q))$, and achieves sample complexity $O(\sqrt{\log(q)\cdot\log(1/\delta)}/(\epsilon \alpha))$. It operates similarly to the Gaussian mechanism, except that the noise samples added to different answers are not independent, but correlated via an appropriate Gaussian process, allowing for much less noise per query. 
Namely, the mechanism returns $a_i=KDE_X(y_i)+Z_i$, where $Z_i\sim N(0,2\log(2/\delta)/(\epsilon n)^2)$ and $\mathrm{Cov}(Z_i,Z_j)=k(y_i,y_j)=e^{-\norm{y_i-y_j}_2^2}$. 
They prove that the mechanism is $(\epsilon,\delta)$-DP for arbitrarily many queries, even though 
the noise magnitude per query does not grow with $q$ at all. (The extra $\sqrt{\log q}$ term in the sample complexity is from a standard bound on the maximum of this finite Gaussian process, ensuring that all $q$ queries are answered accurately simultaneously.) 
The HRW sample complexity is better than all previously mentioned results if approximate DP with sufficiently large $\delta$ (say a small constant $\delta=\Omega(1)$) is allowed. 

\paragraph{Query release vs.~function release.}
The HRW mechanism runs in time linear in $d$ in the query release model, but in order to use it for function release, the curator must release answers to all possible queries, which entails running time exponential in $d$. Thus, in the function release model, to our knowledge, the LSQ-RFF mechanism, despite being pure-DP, is currently the only DP-KDE mechanism for the Gaussian kernel that achieves $(\alpha,\eta)$-approximation with running time linear in $d$, even if approximate DP is allowed.

\subsection{Overview of LSHable Kernels}\label{sec:lshable}
As mentioned in the introduction, the Laplacian kernel $k(x,y)=e^{-\norm{x-y}_1}$ is likely the most popular LSHable kernel over $\R^d$. For completeness, in this section we give an overview of other kernels known to be LSHable. 

\cite{rahimi2007random} introduced a family of LSHable kernels (although they did not use this terminology) in their Random Binning Features construction. They start by showing that the \emph{hat kernel} over $x,y\in\R$, $\hat k_\sigma(x,y)=\max\{0,1-|x-y|/\sigma\}$, is LSHable. 
They then show this implies the LSHability of shift-invariant kernels over $\R$ that can be written as convex combinations of such hat kernels on a compact subsets of $\R\times \R$ (this includes the one-dimensional Laplacian kernel $k(x,y)=e^{-|x-y|}$), and of kernels over $\R^d$ that can be written as the product of one-dimensional LSHable kernels over the coordinates (this includes the $d$-dimensional Laplacian kernel $k(x,y)=e^{-\norm{x-y}_1}=\prod_{i=1}^de^{-|x_i-y_i|}$). They note that this family does not include the Gaussian kernel. 

\cite{andoni2009nearest} discussed additional LSHable kernels over $\R^d$: the \emph{exponential} kernel $k(x,y)=e^{-\norm{x-y}_2}$, whose LSHability follows from that of the Laplacian kernel essentially by an (efficient and approximate) isometric embedding of $\ell_2$ into $\ell_1$;\footnote{Note that the exponential kernel is different from the Laplacian kernel in that the norm in the exponent is $\ell_2$ and not $\ell_1$, and is different from the Gaussian kernel in that the norm is not squared.} the \emph{geodesic} kernel over the unit sphere, $k(x,y)=1-\pi^{-1}\theta(x,y)$, where $\theta(x,y)$ denotes the angle between $x$ and $y$; and the Erfc kernel $k(x,y)=\tfrac{\mathrm{erfc}(\norm{x-y}_2)}{2-\mathrm{erfc}(\norm{x-y}_2)}$, where $\mathrm{erfc}(z)=\tfrac{2}{\sqrt\pi}\int_{z}^\infty e^{-t^2}\mathrm{d}t$ is the complementary Gauss error function. 
Regarding the lack of LSHability results for the Gaussian kernel, they suggest using the Erfc kernel as a proxy (naming it a ``near-Gaussian kernel''), showing it approximates the Gaussian kernel at every up point up to an additive error of $0.16$. 
Unfortunately, this error is far too large for most KDE applications. Furthermore, the LSH family associated with the Erfc kernel has running time that depends exponentially on the additive error $\alpha$ (where $\alpha$ is the approximation error for the Erfc kernel, leading to an error of $0.16+\alpha$ for the Gaussian kernel), making it infeasible when $\alpha$ is small. 

The lack of available LSHability results for the Gaussian and Cauchy kernel is also discussed in \cite{backurs2018efficient,siminelakis2019rehashing}, who develop alternative methods for (non-private) approximation of these kernels where normally LSHability would be used. 

Finally, apart from $\R^d$, some LSHability results are available for kernels that measure similarity over finite spaces. \cite{andoni2009nearest} observe that the Jaccard kernel is LSHable, while \cite{chierichetti2015lsh} discuss transformations that preserve the LSHability of such kernels.

\section{Extensions to Other Kernels}\label{sec:otherkernels}

In this section we discuss the applicability of our results beyond the Gaussian kernel. 
The key distinction to draw here is between the \emph{sample complexity} of the DP-KDE mechanism (i.e., the tradeoff between the privacy parameter $\epsilon$ and the additive error parameters $\alpha,\eta$), for which we can make general statements for some families of kernels, to the \emph{computational efficiency} of the mechanism (i.e. the running times of the curator and the client, and the curator output size), which would generally depend on the specific properties of each kernel. 

\subsection{LSQ with RFF}\label{sec:otherkernels_rff}
\cite{rahimi2007random} showed that every positive definite shift-invariant kernel (abbreviated henceforth as a PDSI kernel) admits a family of random Fourier features. More precisely, for every such kernel $k$ defined over $\R^d$, there exists a distribution $\mathcal D^{RFF}_k$ over $\R^d$ such that
\[ \forall x,y\in\R^d \;\; , \;\; k(x,y) = \E_{\omega\sim\mathcal D^{RFF}_k,\beta\sim\mathrm{Uniform}[0,2\pi)}[\sqrt{2}\cos(\omega^Tx+\beta) \cdot \sqrt{2}\cos(\omega^Ty+\beta)] . \]
This implies that every PDSI kernel is $(1,\sqrt{2},1)$-LSQable. Therefore, from \Cref{lmm:privacy,lmm:utility_largeq} we get the following result.

\begin{theorem}\label{thm:lsqrff}
For every PDSI kernel over $\R^d$, there is an $\epsilon$-DP function release mechanism for $(\alpha,\eta)$-approximation of its KDE, on datasets of size at least $n\geq O(\log(1/\eta)/(\epsilon\alpha^2))$.
\end{theorem}

These are the same privacy, utility and sample complexity guarantees as we get for the Gaussian kernel in \Cref{thm:gauhighdim}. However, the computational efficiency (and more specifically in the case, the curator running time) depends on the computational properties of $D^{RFF}_k$ for each specific kernel $k$. Namely, it hinges on whether one can sample $\omega$ from $\mathcal D^{RFF}_k$ efficiently. Formally, by \Cref{lmm:efficiency}, we get:
\begin{proposition}\label{prp:lsqrff_efficiency}
Let $k$ be a PDSI kernel. 
Let $T_k^{RFF}$ be the time complexity of drawing a sample $\omega$ from $\mathcal D^{RFF}_k$. Then, the LSQ-RFF DP-KDE mechanism from \Cref{thm:lsqrff} satisfies the following:
\begin{itemize}
  \item The curator runs in time $O((nd+T_k^{RFF})\log(1/\eta)/\alpha^2)$.
  \item The curator output size is $O(d\log(1/\eta)/\alpha^2)$.
  \item The client runs in time $O(d\log(1/\eta)/\alpha^2)$.
\end{itemize}
\end{proposition}
\begin{proof}
In the notation of \Cref{lmm:efficiency}, we have $T_{\mathcal Q}=T_k^{RFF}$. Furthermore, $L_{\mathcal Q}=d+1$ since this is the number of machine words needed to describe a pair $\omega,\beta$ (regardless of the time it took to sample $\omega$), and $T_f=T_g=O(d)$ since computing $\sqrt{2}\cos(\omega^Tx+\beta)$ given $x,\omega,\beta$ takes time $O(d)$. We plug these into \Cref{lmm:efficiency} together with $I=O(\log(1/\eta)/\alpha^2)$, the setting of $I$ used in \Cref{lmm:utility_largeq} to obtain \Cref{thm:lsqrff}, and the proposition follows.
\end{proof}

Let us give some examples of $\mathcal D^{RFF}_k$ and $T_k^{RFF}$ for specific kernels, and observe how they affect the efficiency of the LSQ-RFF mechanism.
\begin{itemize}
  \item \emph{Gaussian, Laplacian and Cauchy kernels:} For these three kernels, mentioned in \Cref{sec:kde}, \cite{rahimi2007random} derived the corresponding RFF distributions (we list them here with bandwidth $\sigma=1$):
  \begin{itemize}
    \item For the Gaussian kernel $k(x,y)=\exp(-\norm{x-y}_2^2)$, $\mathcal D^{RFF}_k$ is the $d$-dimensional Gaussian distribution $\sqrt{2}\cdot N(0,I_d)$. 
    \item For the Laplacian kernel $k(x,y)=\exp(-\norm{x-y}_1)$, $\mathcal D^{RFF}_k$ is the $d$-dimensional Cauchy distribution, whose density at $\omega\in\R^d$ is $\prod_{j=1}^d(\pi(1+\omega_j^2))^{-1}$.
    \item For the Cauchy kernel $k(x,y)=\prod_{j=1}^d2/(1+(x_j-y_j)^2)$, $\mathcal D^{RFF}_k$ is the $d$-dimensional Laplace distribution $\mathrm{Laplace}(0, I_d)$.
  \end{itemize}
  Each of these distributions is a $d$-dimensional product distribution where each coordinate can be sampled in time $O(1)$, hence $T_k^{RFF}=O(d)$. Therefore, for these kernels, we get the same DP-KDE results as stated for the Gaussian kernel in \Cref{thm:gauhighdim}.
  \item \emph{Exponential $\ell_p^p$ kernels:} Let $p\in[1,2]$. Consider the kernel $k(x,y)=\exp(-\norm{x-y}_p^p)$. This can be seen as a generalization of the Gaussian and Laplacian kernels (which correspond to $p=2$ and $p=1$ respectively). For this kernel, it can be checked that $\mathcal D^{RFF}_k$ is the $d$-dimensional product distribution whose coordinates are i.i.d.~samples from the $p$-stable distribution, and furthermore, each coordinate can be sampled in time $O(1)$. See \cite{indyk2006stable} for the definition of the $p$-stable distribution and for how to efficiently sample from it. Therefore, for these kernels too we have $T_k^{RFF}=O(d)$, and we get the same DP-KDE result as in \Cref{thm:gauhighdim}.
  \item \emph{Exponential $\ell_p$ kernels:} Again let $p\in[1,2]$, and consider the kernel $k(x,y)=\exp(-\norm{x-y}_p)$. Note that, in contrast to the previous case, the $\ell_p$-norm in the exponent is not raised to the power $p$. The $p=1$ case again coincides with the Laplacian kernel, while the $p=2$ case coincides with the exponential kernel mentioned in \Cref{sec:lshable}. These kernels are PDSI, hence \Cref{thm:lsqrff,prp:lsqrff_efficiency} hold for them. However, we do not immediately see how to efficiently sample from their RFF distribution $\mathcal D^{RFF}_k$ (even though it may be possible), and are therefore unable to determine $T_k^{RFF}$ and bound the curator running time of their LSQ-RFF DP-KDE mechanism. 
\end{itemize}

\subsection{LSQ with FGT}
The Fast Gauss Transform is rather specialized to the Gaussian kernel. Nonetheless, it can be extended to certain kernels with sufficiently similar properties, like those discussed in \cite{alman2020algorithms}, section 9.3. For those kernels, we get the same DP-KDE results as we get for the Gaussian kernel in \Cref{thm:gaulowdim}.

\subsection{LSQ with LSH}
With LSH, the situation is similar to LSQ-RFF: for every LSHable kernel we can get a DP-KDE mechanism with the same privacy and utility guarantees as \Cref{thm:lsqrff}, but the computational efficiency depends on the properties of the LSH family associated with that specific kernel. More precisely, we have the following result, which we recall follows already from the prior work of \cite{coleman2020one}.

\begin{theorem}\label{thm:lsqlsh}
For every $\alpha$-approximate LSHable kernel over $\R^d$, there is an $\epsilon$-DP function release mechanism for $(\alpha,\eta)$-approximation of its KDE, on datasets of size at least $n\geq O(\log(1/\eta)/(\epsilon\alpha^2))$.
\end{theorem}
\begin{proof}
By \Cref{prp:lsh2lsq}, $k$ is $2\alpha$-approximate $(\lceil1/\alpha\rceil,1,1)$-LSQable, hence the theorem follows from \Cref{lmm:privacy,lmm:utility_largeq}.
\end{proof}
These are the same privacy, utility and sample complexity guarantees as we get for the Gaussian kernel in \Cref{thm:lsqrff} (however, note that PDSI kernels and LSHable kernels are distinct classes of kernels). 
The computational efficiency of the LSH-based mechanism depends on the computational properties of the LSH family as follows.
\begin{proposition}\label{prp:lsqrff_efficiency}
Let $k$ be an $\alpha$-approximate LSHable kernel over $\R^d$.
Let $\mathcal H$ be the associated LSH family. 
Let $B$ be range size (i.e., number of hash buckets) of the hash functions in $\mathcal H$. Let $T_{\mathcal H}$ be the time to sample $h\sim \mathcal H$, let $T_h$ be the time to evaluate $h(x)$ given $h\in\mathcal H$ and $x\in\R^D$, and let $L_{\mathcal H}$ be the description size of $h\in\mathcal H$. 
Then, the LSQ-LSH DP-KDE mechanism from \Cref{thm:lsqlsh} satisfies the following:
\begin{itemize}
  \item The curator runs in time $O((nT_h + T_{\mathcal H})\log(1/\eta)/\alpha^2)$.
  \item The curator output size is $O((L_{\mathcal H} + \min\{B,\log B + 1/\alpha\})\cdot\log(1/\eta)/\alpha^2)$.
  \item The client runs in time $O(T_h\log(1/\eta)/\alpha^2)$.
\end{itemize}
\end{proposition}
\begin{proof}
Let $I=O(\log(1/\eta)/\alpha^2)$, noting this is the setting of $I$ used in \Cref{lmm:utility_largeq} to obtain \Cref{thm:lsqlsh}.

Recall that we have two options to transform the LSH family into an LSQ family: either by  \Cref{prp:lsh2lsq_b} or by \Cref{prp:lsh2lsq}. We analyze both cases. 
If we use \Cref{prp:lsh2lsq_b}, then $k$ is $(B,1,1)$-LSQable, and in the notation of \Cref{lmm:efficiency} we have $T_{\mathcal Q}=T_{\mathcal H}$, $T_f=T_g=T_h$, and $L_{\mathcal Q}=L_{\mathcal H}$. Applying \Cref{lmm:efficiency}, the curator running time is $O(I(nT_h + T_{\mathcal H}+B))$, the curator output size is $O(I(L_{\mathcal H} + B))$, and the client running in time $O(I\cdot T_h)$.

Alternatively, if we use \Cref{prp:lsh2lsq}, then $k$ is $(\lceil1/\alpha\rceil,1,1)$-LSQable. The proof of \Cref{prp:lsh2lsq} (cf.~\Cref{{sec:lshappendix}}) obtains the LSQ family by composing over $\mathcal H$ a universal hash family $\mathcal U$ that hashes a domain of size $B$ into $\lceil1/\alpha\rceil$ hash buckets. There are well-known choices for $\mathcal U$ (e.g., \cite{carter1977universal}) with sampling and evaluation times $O(1)$ and description size $O(\log B)$. Hence, for the composition of $\mathcal U$ over $\mathcal H$, we have in the notation of \Cref{lmm:efficiency} $T_{\mathcal Q}=T_{\mathcal H}+O(1)$, $T_f=T_g=T_h+O(1)$, and $L_{\mathcal Q}=L_{\mathcal H}+O(\log B)$. Applying \Cref{lmm:efficiency}, the curator running time is $O(I(nT_h + T_{\mathcal H}+1/\alpha))$, the curator output size is $O(I(L_{\mathcal H} + \log B + 1/\alpha))$, and the client running in time $O(I\cdot T_h)$.

Putting these together, the curator running time is $O(I(nT_h + T_{\mathcal H}+\min\{B,1/\alpha\}))$, the curator output size is $O(I(L_{\mathcal H} + \min\{B,\log B + 1/\alpha\}))$, and the client running in time $O(I\cdot T_h)$. Note that the approximation guarantee in \Cref{thm:lsqlsh} requires $n\geq O(\log(1/\eta)/(\epsilon\alpha^2))$, hence $nT_h \geq n \geq O(1/\alpha)$, and hence the curator running time becomes $O(I(nT_h + T_{\mathcal H}))$. These are the bounds claimed in the proposition.
\end{proof}

Here too, let us give some examples of how different LSH families affect the computational efficiency of the DP-KDE mechanism.
\begin{itemize}
  \item \emph{Laplacian, exponential and geodesic kernels:} as already mentioned in \Cref{sec:lsh}, the Laplacian kernel admits an LSH family that satisfies $T_{\mathcal Q},T_f,T_g,L_{\mathcal Q}=O(d)$ in the notation of \Cref{lmm:efficiency}. Therefore, we get an efficient DP-KDE mechanism for it, as stated in \Cref{thm:laplaciankde}. The exponential kernel and the geodesic kernel, mentioned as LSHable in \Cref{sec:lshable}, also have LSH families with similar (though perhaps slightly different) efficiency properties, given in \cite{andoni2009nearest}.
  \item \emph{Erfc kernel:} In \Cref{sec:lshable} we defined the Erfc kernel, and mentioned that \cite{andoni2009nearest} showed it is $\alpha$-approximate LSHable, albeit with an LSH family that takes time exponential in $\alpha$ to sample from. Therefore, for this kernel we get a DP-KDE mechanism with the privacy, utility and sample complexity stated in \Cref{thm:lsqlsh}, but with running time exponential in $\alpha$.
  \item \emph{Exponential $\ell_p$ kernels:} Let us revisit the family of kernels $k(x,y)=\exp(-\norm{x-y}_p)$ with $p\in[1,2]$. We discussed these kernels in the context of LSQ-RFF, and showed that while we have DP-KDE mechanisms for them, we do not know them to be computationally efficient. This result also follows by LSHability. The reason is that $\ell_p$ is known to embed isometrically into $\ell_1$ \cite{johnson1982embedding}. This implies that the kernel $k(x,y)=\exp(-\norm{x-y}_p)$ with any $p\in[1,2]$ is LSHable, by first applying an isometric embedding of the $\ell_p$ distances into $\ell_1$, and then using the LSHability of the Laplacian kernel. However, except in the $p=2$ case, it is not known how to compute an (approximately) isometric embedding of $\ell_p$ into $\ell_1$ efficiently. Therefore, while for these kernels we can get DP-KDE mechanisms from \Cref{thm:lsqlsh}, we are unable to bound their computational efficiency.
\end{itemize}

\section{Additional Experiments and Implementation Details}\label{sec:experiments_appendix}

\subsection{Mechanism Implementation}\label{sec:mechanism_implementation}
In this section we provide details on how we instantiate the LSQ mechanism from \Cref{alg:main} into the LSQ-RFF and LSQ-FGT mechanisms included in our code and used in our experiments, and on how these mechanisms are parameterized.

The efficiency/utility trade-off of the LSQ mechanism in \Cref{alg:main} is governed by the input parameters $I,J$, which are non-negative integers such that $J$ is a divisor of $I$. (Observe that the computational efficiency bounds in \Cref{lmm:efficiency} grow linearly with $I$.) Their role is simply to determine the number of repetitions in a standard median-of-means (MoM) averaging scheme, to induce the desired probabilistic concentration. The mechanism performs a total of $I$ independent repetitions, and uses them to return the median of $J$ terms, where each term is the average of $I'=I/J$ repetitions. As usual with MoM, $I'$ governs the additive error $\alpha$ that we consider ``successful'', while $J$ governs the probability $\eta$ of failing to achieve that successful additive error.

From a typical theoretical perspective, one would like to select the desired utility parameters $\alpha$ and $\eta$, and ensure that the mechanism rigorously satisfies $(\alpha,\eta)$-approximation. To this end, \Cref{lmm:utility_smallq,lmm:utility_largeq} specify the setting of $I$ and $J$ that formally guarantees $(\alpha,\eta)$-approximation and leads to our theoretical results, \Cref{thm:gauhighdim,thm:gaulowdim}. 

For our experiments, however, we would like to directly control the computational cost of our mechanisms, and measure their empirical utility as we vary the computational cost. 
To this end, we parameterize each of our two implemented mechanisms---LSQ-RFF and LSQ-FGT---by a single parameter that governs their computational efficiency, as follows. 
In both mechanisms, for simplicity, we use $J=1$, which means we do not perform a median operation at all. One can always increase $J$ and return the median over $J$ independent repetitions in order to boost the success probability of each individual query, at the expense of degrading $\epsilon$ (by a factor of $J$) for releasing more information in those additional repetitions.

In LSQ-RFF, we parameterize the mechanism by the number of random Fourier features the mechanism uses, which (under the setting $J=1$) coincides with the overall number of repetitions, $I$, in \Cref{alg:main}.

In LSQ-FGT, there is the added complication that the LSQ family itself has variable computational cost. In order to define the FGT, the user selects an integer parameter $\rho\geq 1$, which determines the properties of the LSQ family as follows:
\begin{proposition}
Let $\rho\geq1$ be an integer.
The Gaussian kernel over points contained in a Euclidean ball of radius $\Phi$ in $\R^d$ admits an $e^{-O(\rho)}$-approximate $((1+\frac{\Phi}{\sqrt{d}})\cdot\rho)^d, O(1)^d, \rho^{O(d)})$-LSQ family, supported on a single pair of functions $(f,g)$. 
Furthermore, the evaluation times of $f$ on $x\in\R^d$ and of $g$ on $y\in\R^d$ are both $(d\cdot\rho)^{O(d)}$. 
\end{proposition}
This is just a restatement of \Cref{prp:fgt}, parameterized by $\rho$ instead of $\alpha$ (and it follows from the same proof in \Cref{sec:fgtappendix}). 
Note that as $\rho$ increases, the parameters $Q$ and $S$ of the $(Q,R,S)$-LSQ family grow with it, which increases the computational cost of the LSQ mechanism according to \Cref{lmm:efficiency}. 
The description of LSQ-FGT in \Cref{sec:fgt} sets $\rho=O(\log(1/\alpha))$ in order to prove \Cref{thm:gaulowdim}, but in practice, when $\alpha$ is not chosen in advance but measured empirically, the user needs to set $\rho$ directly. In our implementation of LSQ-FGT, we set the number of repetitions to $I=1$, and use $\rho$ as the parameter that governs the efficiency/utility trade-off.

\subsection{Experimental Details}\label{sec:expdetails}

\paragraph{Preprocessing.} All datasets are available online (download URLs are included in the bibliographic entries).
\begin{itemize}
  \item Covertype \cite{blackard1999comparative}: No preprocessing.
  \item GloVe \cite{pennington2014glove}: We use the 1M points, 100 dimensions version of the dataset. No preprocessing.
  \item Diabetes \cite{strack2014impact}: we select the ``age'' and ``time in hospital'' columns. ``time in hospital'' is between $1$ and $14$ (days). ``age'' is given as a decade-long bracket (e.g., $[40-50)$) and we replace it with its midpoint (e.g., $45$), and then divide it by $10$ to equate the numerical range of both coordinates.\footnote{This is equivalent to choosing the bandwidth as a non-scalar diagonal matrix, namely $\begin{pmatrix}
1 & 0\\
0 & 0.1
\end{pmatrix}$. Recall that the bandwidth is, in general, a $d\times d$ positive definite matrix $\Sigma$, with which the Gaussian kernel is defined as $k(x,y)=e^{-(x-y)^T\Sigma(x-y)}$.}
  \item NYC Taxi \cite{new-york-city-taxi-fare-prediction}: We select the ``pickup longitude'' and ``pickup latitude'' columns. We filter out points with ``pickup longitude'' $\notin(-74.1, -73.15)$ or ``pickup latitude'' $\notin(40.5, 40.9)$ to eliminate corrupted records (these coordinate ranges are the general geographical vicinity of NYC). We use $100,000$ of the unfiltered points.
\end{itemize}

\paragraph{Bandwidth selection.}
For each dataset we tune the bandwidth according to the guidelines in prior work \cite{jaakkola1999using,backurs2019space}. 
The values are specified in \Cref{tbl:bandwidth}. 
The bandwidth values are tuned are such that mean KDE values are on the order of $10^{-2}$ and their standard deviation is also on the order of $10^{-2}$, yielding a meaningful and non-generate KDE distribution with a range of target values. Note that the performance of the NoisySample baseline in \Cref{sec:experiments} (which returns the noisy mean of a sample of query points as the KDE estimate for any query point) corresponds to the standard deviation of KDE values in \Cref{tbl:bandwidth}. 

\begin{table*}[ht]
\caption{Bandwidth values used in experiments.}
{\renewcommand{\arraystretch}{1}
\begin{center}
\begin{tabular}{lccc}
\toprule
  Dataset &  Bandwidth $\sigma$ &  Est.~mean query KDE &  Est.~standard deviation of query KDE \\
\midrule
Covertype & $500$ & $0.02$ & $0.01$ \\
GloVe & $3.33$ & $0.01$ & $0.01$ \\
Diabetes & $1$ & $0.06$ & $0.03$ \\
NYC Taxi & $0.01$ & $0.08$ & $0.03$ \\
\bottomrule
\end{tabular}
\end{center}}
\label{tbl:bandwidth}
\end{table*}

\paragraph{Mechanism parameter selection.}
As discussed in \Cref{sec:error_progression}, DP-KDE mechanisms have an optimal parameter setting for a given combination of error $\alpha$ and privacy $\epsilon$. 
In our experiments this applies to LSQ-RFF (the parameter is the number of Fourier features), LSQ-FGT (the parameter is $\rho$, where $\rho^d$ is the number of terms in truncated Hermite expansion) and Bernstein (the parameter is denote by $k$ in \cite{alda2017bernstein}, where $(k+1)^d$ is the number of points in the lattice used to construct the Bernstein polynomial approximator, see below). 
In the error vs.~privacy experiments in \Cref{sec:error_privacy}, we evaluate each mechanism at its optimal parameter for that specific value of $\epsilon$. 
Due to the existence of the error divergence point (cf.~\Cref{sec:error_progression}), the optimal parameter setting for each algorithm exists and can be found by a finite parameter search.

For completeness, let us describe the Bernstein mechanism is somewhat more detail. It is parameterized by an integer $k\geq1$. The mechanism constructs a uniform lattice with $(k+1)^d$ nodes over the unit hypercube $[0,1]^d$. It evaluates the KDE function at each point on the lattice, adds privacy-preserving Laplace noise to these evaluations, and then uses them to construct a Bernstein polynomial approximation of this discretized and privatized version of the true KDE function. As $k$ increases, the mechanism's running time increases too, due to evaluating the KDE on each of the $(k+1)^d$ lattice points. 
Nonetheless, as shown for LSQ-RFF and LSQ-FGT in \cref{sec:error_progression}, increasing $k$ does not necessarily lead to a smaller error---rather, the error begins to diverge at a certain setting of $k$, which depends on the desried privacy parameter $\epsilon$. This happens for the same reason discussed in \cref{sec:error_progression}: as $k$ increases, the non-private approximation error of the Bernstein polynomial approximator decays (see Theorem 5 in \cite{alda2017bernstein} for the decay rate, which depends on the smoothness of the KDE function), while the magnitude of the Laplace noise increases like $(k+1)^d/(\epsilon n)$. Therefore, to achieve the optimal error for this mechanism, $k$ needs to be chosen according to the available dataset size $n$ and the desired privacy level $\epsilon$. 

\subsection{Additional Accuracy Results}
\label{sec:scatter_plots}
A more visual way to study the privacy-error trade-off of the various  DP-KDE mechanisms is by directly comparing the ground-truth KDE values on a held-out test set to the KDE values estimated by the private mechanisms for different values of $\epsilon$.
\Cref{fig:scatter_all} shows the performance of LSQ-RFF under varying privacy budgets for the high-dimensional Covertype and GloVe datasets. 
Ideally, the estimated values would all lie close to the $y=x$ line, but degradation is inevitable as $\epsilon$ decreases. 
Additionally, \Cref{fig:scatter_all} compares the performance LSQ-RFF, LSQ-FGT, and the Bernstein mechanism on the low-dimensional Taxi and Diabetes datasets under the same set of privacy budgets. In~\Cref{sec:error_privacy}, we noted that the NYC Taxi dataset poses a challenge for the Bernstein mechanism because of the dependence of sample complexity on $\alpha^{-\Theta(d/\sigma^2)}$. This difficulty manifests itself already in the non-private case, and (as expected) the mechanism output quality degrades further once noise is introduced to preserve privacy.


\subsection{Additional Running Time Results}\label{sec:appruntimes}
In \Cref{fig:runtime} in \Cref{sec:experiments} we plotted the error vs.~curator running time plots for all for our datasets, with $\epsilon=0.05$. 
\Cref{,fig:runtime002} below displays the same experiment with $\epsilon=0.02$. 

\subsection{Heatmaps}
\label{sec:heatmaps}
A common use for Kernel Density Estimation for two-dimensional datasets is in the generation of heatmaps showing where the bulk of the samples reside. For both the NYC Taxi and the Diabetes dataset, we use LSQ-RFF and LSQ-FGT to generate differentially private heatmaps for a number of different privacy budgets.  The parameters of these algorithms (number of features for RFF, $\rho$ for FGT) are selected to match the optimal values found earlier in~\Cref{sec:error_progression}. 
Results are in \Cref{fig:heatmaps}.  In all cases, while the heatmap gets increasingly distorted as the privacy budget shrinks, certain aggregate characteristics such as the general shape of the data manifold and the approximate location of its mode remain largely preserved.

\begin{figure*}[ht]
\centering
\includegraphics[width=0.48\columnwidth]{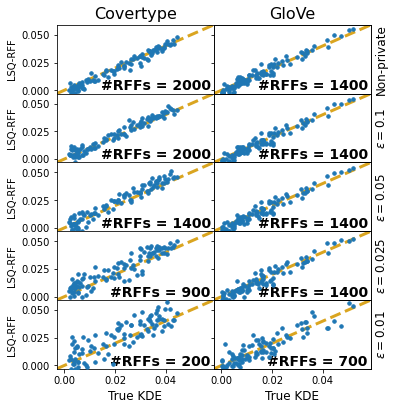}
\begin{center}
    {Covertype and GloVe datasets}
\end{center}
\vspace{20pt}
%
\centering
\begin{minipage}[t]{0.48\columnwidth}
\begin{minipage}[b]{\columnwidth}
    \includegraphics[width=\linewidth]{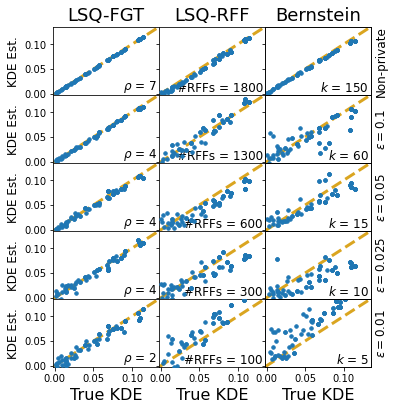}
\begin{center}
    {Diabetes dataset}
\end{center}
\end{minipage}
\end{minipage}\hfill
\begin{minipage}[t]{0.48\columnwidth}
\begin{minipage}[b]{\columnwidth}
    \includegraphics[width=\linewidth]{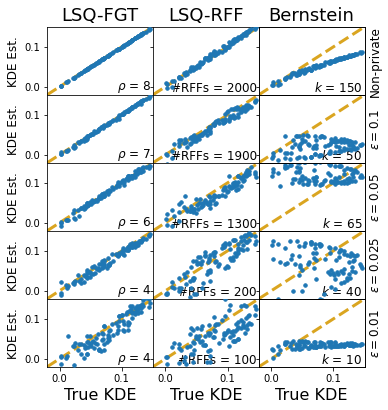}
\begin{center}
    {NYC Taxi dataset}
\end{center}
\end{minipage}
\end{minipage}
\caption{Ground truth vs.~private estimates}
\label{fig:scatter_all}
\end{figure*}

\begin{figure*}[ht]
\centering
\begin{minipage}[t]{\columnwidth}
\begin{minipage}[b]{\columnwidth}
\includegraphics[width=0.45\linewidth]{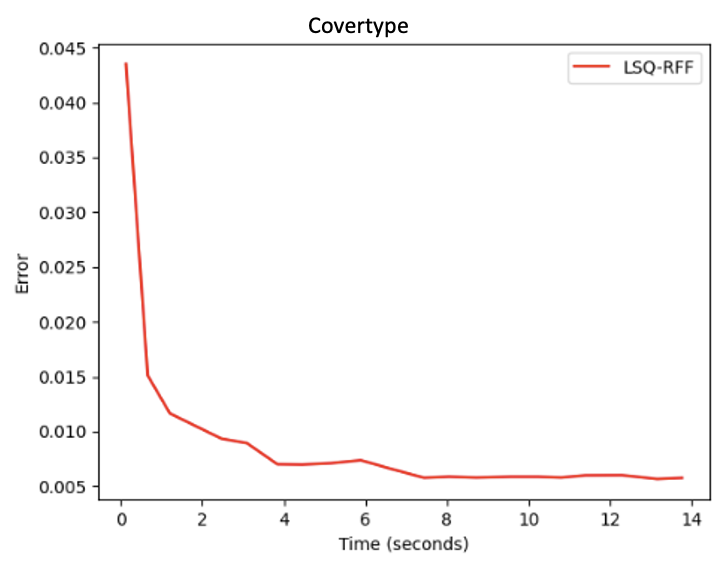}%
\includegraphics[width=0.45\linewidth]{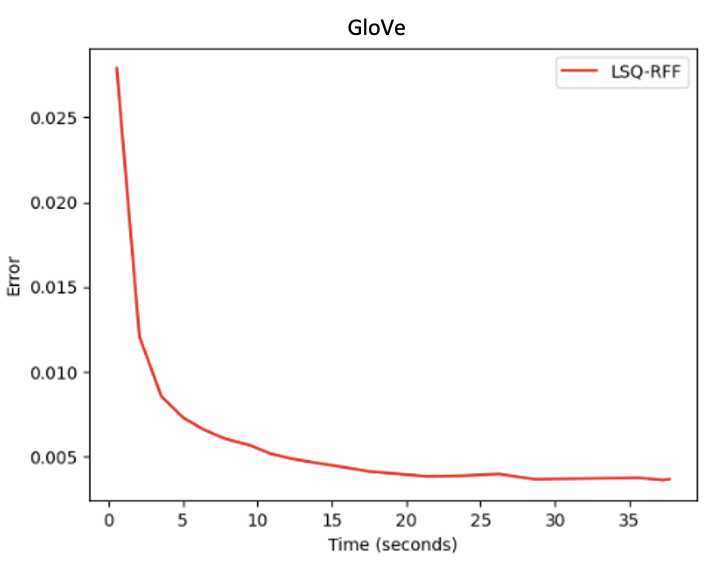}
\end{minipage}
\end{minipage}
\begin{minipage}[t]{\columnwidth}
\begin{minipage}[b]{\columnwidth}
\includegraphics[width=0.45\linewidth]{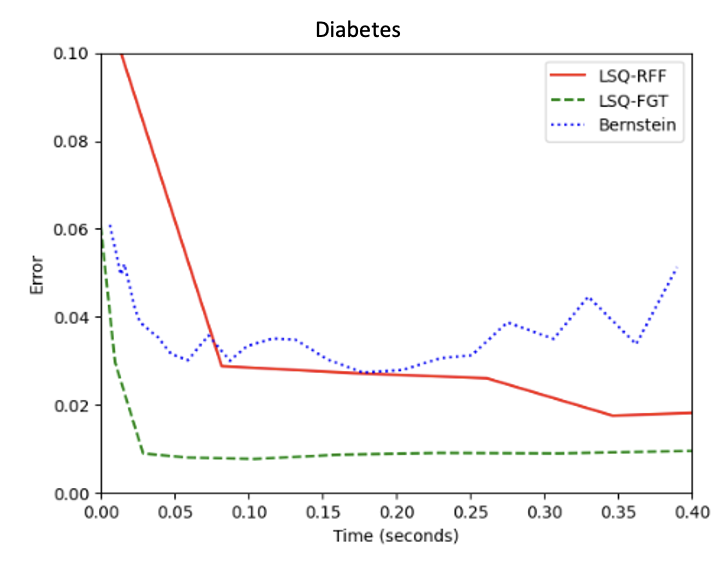}%
\includegraphics[width=0.45\linewidth]{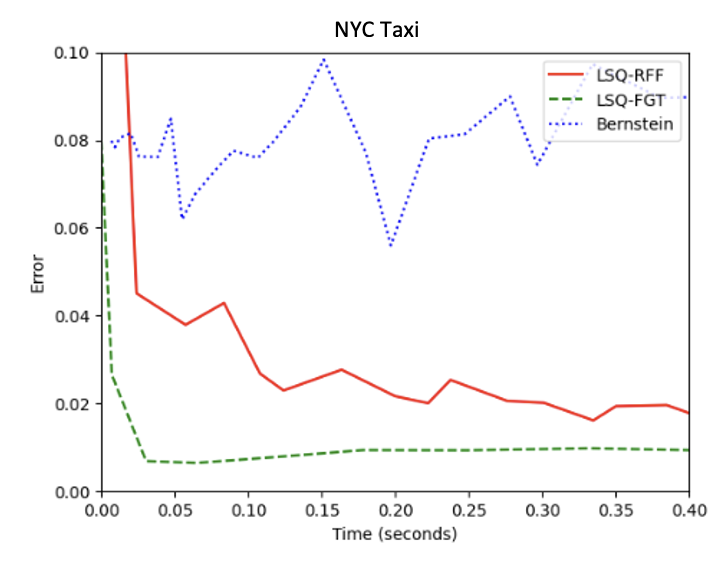}
\end{minipage}
\end{minipage}%
\vspace{-10pt}
\caption{Error vs. curator running times with $\epsilon=0.02$}
\label{fig:runtime002}
\end{figure*}

\begin{figure*}[ht]
\centering
\begin{minipage}[t]{0.48\linewidth}
\begin{minipage}[b]{\linewidth}
    \includegraphics[width=\linewidth]{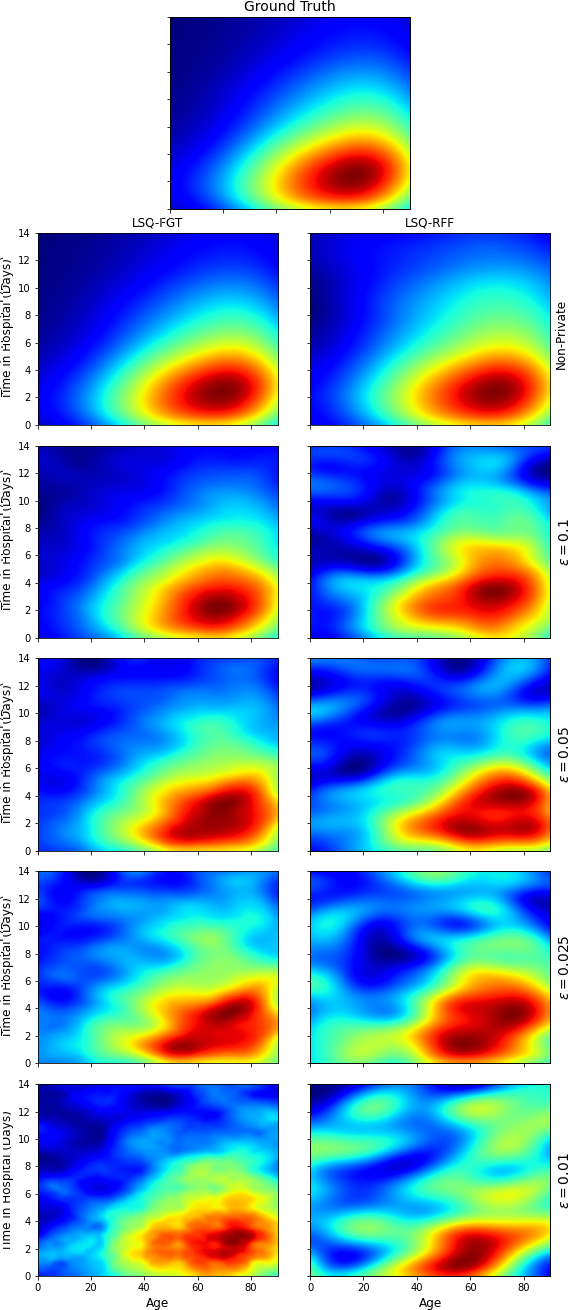}
\begin{center}
    {Diabetes dataset}
\end{center}
\end{minipage}
\end{minipage}\hfill
\begin{minipage}[t]{0.48\linewidth}
\begin{minipage}[b]{\linewidth}
\includegraphics[width=\linewidth]{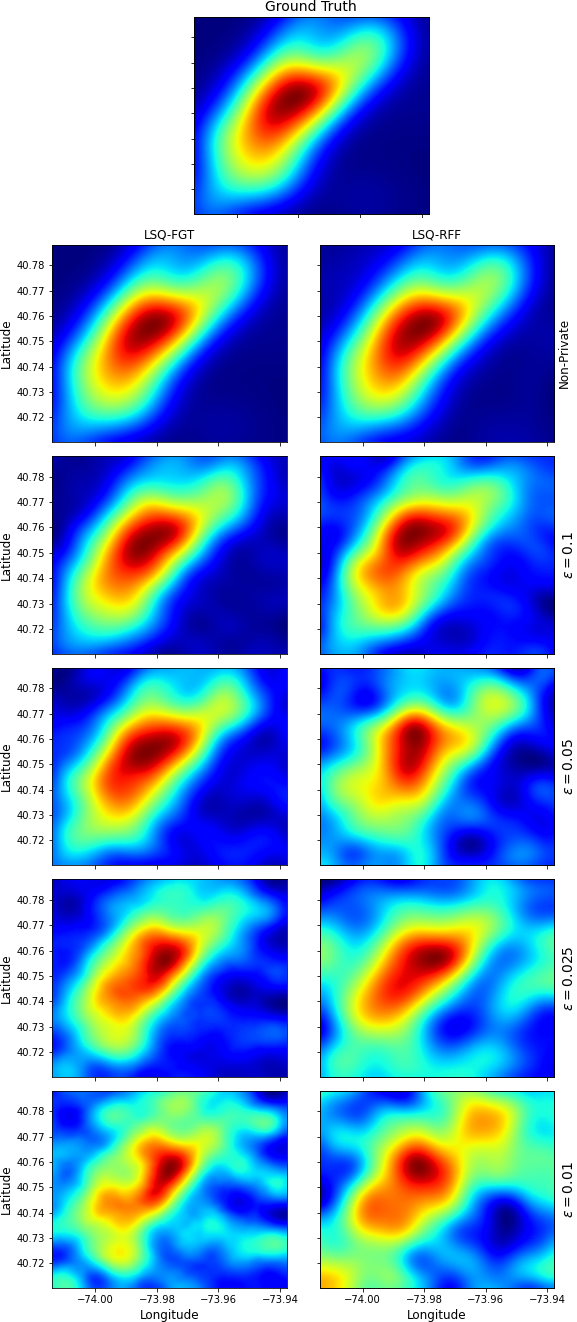}
\begin{center}
    {NYC Taxi dataset}
\end{center}
\end{minipage}
\end{minipage}
\caption{Impact of privacy budget on the appearance of heatmap plots}
\label{fig:heatmaps}
\end{figure*}

\end{document}